\documentclass[a4paper,UKenglish,cleveref, autoref, thm-restate,authorcolumns]{lipics-v2019}
\graphicspath{{./supplement/}}%helpful if your graphic files are in another directory

\title{Algorithms for new types of fair stable matchings} %TODO Please add

\titlerunning{Algorithms for new types of fair stable matchings} %TODO optional, please use if title is longer than one line

\author{Frances Cooper}{School of Computing Science, University of Glasgow, Glasgow, Scotland, UK \and \url{https://www.francescooper.net}}{f.cooper.1@research.gla.ac.uk}{https://orcid.org/0000-0001-6363-9002}{Supported by an Engineering and Physical Sciences Research Council Doctoral Training Account EP/N509668/1}%TODO mandatory, please use full name; only 1 author per \author macro; first two parameters are mandatory, other parameters can be empty. Please provide at least the name of the affiliation and the country. The full address is optional

\author{David Manlove}{School of Computing Science, University of Glasgow, Glasgow, Scotland, UK \and \url{http://www.dcs.gla.ac.uk/~davidm}}{david.manlove@glasgow.ac.uk}{https://orcid.org/0000-0001-6754-7308}{Supported by Engineering and Physical Sciences Research Council grant EP/P028306/1}

\authorrunning{F. Cooper and D. Manlove} %TODO mandatory. First: Use abbreviated first/middle names. Second (only in severe cases): Use first author plus 'et al.'

\Copyright{Frances Cooper and David Manlove} %TODO mandatory, please use full first names. LIPIcs license is "CC-BY";  http://creativecommons.org/licenses/by/3.0/

 %TODO mandatory: Please choose ACM 2012 classifications from https://dl.acm.org/ccs/ccs_flat.cfm 

 %TODO mandatory; please add comma-separated list of keywords

\category{} %optional, e.g. invited paper

\relatedversion{\href{https://arxiv.org/abs/2001.10875}{arxiv.org/abs/2001.10875}} %optional, e.g. full version hosted on arXiv, HAL, or other respository/website
\supplement{\href{https://zenodo.org/record/3630383}{zenodo.org/record/3630383} and \href{https://zenodo.org/record/3630349}{zenodo.org/record/3630349}}%optional, e.g. related research data, source code, ... hosted on a repository like zenodo, figshare, GitHub, ...

\nolinenumbers %uncomment to disable line numbering

\EventEditors{John Q. Open and Joan R. Access}
\EventNoEds{2}
\EventLongTitle{42nd Conference on Very Important Topics (CVIT 2016)}
\EventShortTitle{CVIT 2016}
\EventAcronym{CVIT}
\EventYear{2016}
\EventDate{December 24--27, 2016}
\EventLocation{Little Whinging, United Kingdom}
\EventLogo{}
\SeriesVolume{42}
\ArticleNo{23}
\usepackage[nomain, acronym]{glossaries}
\usepackage{pdflscape}
\usepackage{algorithm}
\usepackage{algpseudocode}
\usepackage{thm-restate}

\usepackage{placeins}

\newcolumntype{L}[1]{>{\raggedright\let\newline\\\arraybackslash\hspace{0pt}}m{#1}}
\newcolumntype{C}[1]{>{\centering\let\newline\\\arraybackslash\hspace{0pt}}m{#1}}
\newcolumntype{R}[1]{>{\raggedleft\let\newline\\\arraybackslash\hspace{0pt}}m{#1}}

\newacronym{sm}{\sc sm}{Stable Marriage problem}
\newacronym{smi}{\sc smi}{Stable Marriage problem with Incomplete lists}
\newacronym{hr}{\sc hr}{Hospitals/Residents problem}
\newacronym{hr-gr}{\sc hr-gr}{Hospitals/Residents problem with Grouped Residents}
\newacronym{smi-gw}{\sc smi-gw}{Stable Marriage problem with Incomplete lists and Grouped Women}
\newacronym{nrmp}{NRMP}{National Resident Matching Program}
\newacronym{ip}{IP}{Integer Programming}

\usepackage[colorinlistoftodos]{todonotes}

\hideLIPIcs
\EventEditors{Simone Faro and Domenico Cantone}
\EventNoEds{2}
\EventLongTitle{18th International Symposium on Experimental Algorithms (SEA 2020)}
\EventShortTitle{SEA 2020}
\EventAcronym{SEA}
\EventYear{2020}
\EventDate{June 16--18, 2020}
\EventLocation{Catania, Italy}
\EventLogo{}
\SeriesVolume{160}
\ArticleNo{17}
\begin{document}

\maketitle

\begin{abstract}
We study the problem of finding ``fair'' stable matchings in the \emph{\acrlong{smi}} (\acrshort{smi}). For an instance $I$ of \acrshort{smi} there may be many stable matchings, providing significantly different outcomes for the sets of men and women. We introduce two new notions of fairness in \acrshort{smi}. Firstly, a \emph{regret-equal stable matching} minimises the difference in ranks of a worst-off man and a worst-off woman, among all stable matchings. Secondly, a \emph{min-regret sum stable matching} minimises the sum of ranks of a worst-off man and a worst-off woman, among all stable matchings. We present two new efficient algorithms to find stable matchings of these types. Firstly, the \emph{Regret-Equal Degree Iteration Algorithm} finds a regret-equal stable matching in $O(d_0 nm)$ time, where $d_0$ is the absolute difference in ranks between a worst-off man and a worst-off woman in the man-optimal stable matching, $n$ is the number of men or women, and $m$ is the total length of all preference lists. Secondly, the \emph{Min-Regret Sum Algorithm} finds a min-regret sum stable matching in $O(d_s m)$ time, where $d_s$ is the difference in the ranks between a worst-off man in each of the woman-optimal and man-optimal stable matchings. Experiments to compare several types of fair optimal stable matchings were conducted and show that the Regret-Equal Degree Iteration Algorithm produces matchings that are competitive with respect to other fairness objectives. On the other hand, existing types of ``fair'' stable matchings did not provide as close an approximation to regret-equal stable matchings.
\end{abstract}

\section{Introduction}

\subsection{Background}
\label{sm_re_background}
The \acrlong{sm} (\acrshort{sm}) was first introduced by Gale and Shapley \cite{GS62} in their seminal paper ``College Admissions and the Stability of Marriage'', and comprises a set of men and a set of women, where each man has a strict preference over all women and vice versa. A \emph{matching} in this setting is an assignment of men to women such that no man or woman is multiply assigned. A \emph{stable matching} is then a matching in which there is no man-woman pair who would rather be assigned to each other than to their assigned partners. 

In this paper we study an extension of \acrshort{sm}, known as the \acrlong{smi} (\acrshort{smi}). 
% A further extension to \acrshort{smi} is known as the \acrlong{hr} (\acrshort{hr}), in which residents take the place of men, hospitals take the place of women, and hospitals may be multiply assigned up to some predetermined capacity. A \emph{matching} in this scenario is defined as an allocation of residents to hospitals such that residents are not multiply assigned and the number of assignments for each hospital is at or below its capacity. A \emph{stable matching} is then a matching in which there is no resident-hospital pair $(r, h)$, who find each other acceptable, such that $r$ is unassigned or prefers $h$ to their currently assigned hospital, and, $h$ is below its capacity or prefers $r$ to its worst assigned resident. Gale and Shapley \cite{GS62} showed that a stable matching always exists in an instance of \acrshort{hr} and can be found in linear time using the Resident-oriented Gale-Shapley Algorithm or the Hospital-oriented Gale-Shapley Algorithm \cite{GS62}. There may be many stable matchings in a particular instance of \acrshort{hr}. 
An instance $I$ of \acrshort{smi} comprises two sets of agents, men $U=\{m_1, m_2, ..., m_n\}$ and women $W=\{w_1, w_2, ..., w_n\}$. Each man (woman) ranks a subset of women (men) in strict preference order. Let $m$ be the total length of all preference lists. A man $m_i$ finds a woman $w_j$ \emph{acceptable} if $w_j$ appears on $m_i$'s preference list. Similarly, a woman $w_j$ finds a man $m_i$ \emph{acceptable} if $m_i$ appears on $w_j$'s preference list.  A pair $(m_i, w_j)$ is \emph{acceptable} if $m_i$ finds $w_j$ acceptable and $w_j$ finds $m_i$ acceptable. A \emph{matching} $M$ in this context is an assignment of men to women comprising acceptable pairs such that no man or woman is assigned to more than one person. Given a matching $M$, denote by $M(m_i)$ the woman $m_i$ is assigned to in $M$ (or if $m_i$ is unassigned then $M(m_i)$ is undefined); the notation $M(w_j)$ is defined similarly for a woman $w_j$. A pair $(m_i, w_j)$ is a \emph{blocking pair} if 1) $(m_i,w_j)$ is an acceptable pair, 2) $m_i$ is unmatched or prefers $w_j$ to $M(m_i)$, and 3) $w_j$ is unmatched or prefers $m_i$ to $M(w_j)$.
Matching $M$ is \emph{stable} if it has no blocking pair.

In \acrshort{smi}, a stable matching always exists, and may be found in linear time using the Man-oriented Gale-Shapley Algorithm or the Woman-oriented Gale-Shapley Algorithm \cite{GS62}. The Man-oriented Gale-Shapley Algorithm produces the \emph{man-optimal} stable matching, that is, the unique stable matching in which each man is assigned their most-preferred woman in any stable matching. Unfortunately, the man-optimal stable matching is also \emph{woman-pessimal} i.e., each woman is assigned their least-preferred man in any stable matching. Similarly the Woman-oriented Gale-Shapley Algorithm produces the woman-optimal (man-pessimal) stable matching.

Let $I$ be an instance of \acrshort{smi} and $n$ be the number of men or women in $I$. Let $\mathcal{M}$ be the set of all stable matchings in $I$, which may be exponential in size \cite{IL86}. We note that by the \emph{``Rural Hospitals'' Theorem} \cite{GS85}, the same set of men and women are assigned in all stable matchings of $\mathcal{M}$. Thus in order to simplify future descriptions, we are able to use the Man-oriented Gale-Shapley Algorithm to find and remove all unassigned men and women from $I$ prior to any other operation. Without loss of generality, we assume that from this point onwards, all men and women in $I$ are assigned in any stable matching of $I$.

 For an instance of \acrshort{smi}, it is natural to wish to find a stable matching in $\mathcal{M}$ which is in some sense fair for both sets of men and women. The \emph{rank} of $m_i$ with respect to $M$ is defined as the location of $M(m_i)$ on $m_i$'s preference list, and is denoted $\text{rank}(m_i, M(m_i))$. An analogous definition of rank$(w_j,M(w_j))$ holds for a woman $w_j$. We define the \emph{man-degree} $d_U(M)$ of $M$ as the largest rank of all men in $M$, that is, $d_U(M) = \max \{ \text{rank}(m_i, M(m_i)) : m_i \in U \}$. Again an analogous definition of $d_W(M)$ holds for women. Define the \emph{degree pair} of $M$, denoted $d(M)=(a, b)$ as the tuple of man- and woman-degrees in $M$, where $a=d_U(M)$ and $b=d_W(M)$. The \emph{man-cost} $c_U(M)$ of matching $M$ is defined as the sum of ranks of all men, that is, $c_U(M)=\sum_{m_i \in U}{\text{rank}(m_i, M(m_i))}$. A similar definition of $c_W(M)$ holds for women.
 Finally, the \emph{degree} of a matching $M$ is given by $d(M) = \max \{ d_U(M), d_W(M)\}$ and the \emph{cost} of matching $M$ is given by $c(M) = c_U(M) + c_W(M)$.

 We now define four notions of fairness in the \acrshort{smi} context. Given a stable matching $M$, define its \emph{balanced score} to be $\max \{c_U(M),c_W(M)\}$. $M$ is \emph{balanced} \cite{Fed90} if it has minimum balanced score over all stable matchings in $\mathcal{M}$. Feder \cite{Fed90} showed that the problem of finding a balanced stable matching in \acrshort{smi} is NP-hard, although can be approximated within a factor of $2$. This approximation factor was improved to $2 - \frac{1}{l}$, where $l$ is the length of the longest preference list, by Eric McDermid as noted in Manlove \cite[pg. 110]{Man13}. Gupta et al.\ \cite{GRSZ19} showed that a balanced stable matching can be found in $O(f(n)8^t)$ time when parameterised by $t = k - \min\{c_U(M_0), c_W(M_z)\}$, where $f(n)$ is a function polynomial in $n$ and $k$ is the balanced score. The \emph{sex-equal score} of $M$ is defined to be $|c_U(M)-c_W(M)|$. $M$ is \emph{sex-equal} \cite{GI89} if it has minimum sex-equal score over all stable matchings in $\mathcal{M}$. Finding a sex-equal stable matching was shown to be NP-hard by Kato \cite{Kat93}. 
%  This result was later strengthened by McDermid and Irving \cite{MI14} who showed that when parameterised by the sex-equal score, finding a sex-equal stable matching is not in class XP. 
This result was later strengthened by McDermid and Irving \cite{MI14} who showed that, even in the case when preference lists have length at most $3$, the problem of deciding whether there is a stable matching with sex-equal score $0$ is NP-complete.
 Additionally, a polynomial-time algorithm to find a sex-equal stable matching is described for instances in which men have preference lists of length at most $2$ (women's preference lists remaining unbounded) \cite{MI14}. A stable matching $M$ is \emph{egalitarian} \cite{Knu76} if $c(M)$ is minimum over all stable matchings in $\mathcal{M}$, and may be found in $O(m^{1.5})$ time \cite{Fed90}. Finally, a stable matching $M$ is \emph{minimum regret} \cite{Knu76} if $d(M)$ is minimum among all stable matchings in $\mathcal{M}$. It is possible to find a minimum regret stable matching in $O(m)$ time \cite{Gus87}. These definitions of fairness are summarised in Table \ref{group_sm_optimisation_options}.

 \begin{table}[h] \centerline{\begin{tabular}{ p{2.3cm} | p{5.2cm} p{5.2cm} }\hline\hline 
& Cost & Degree \\
\hline
\\[-0.5em]
\multirow{2}{2.4cm}{Minimising the maximum} & $\underset{M \in \mathcal{M}}{\min} \max \{c_U(M), c_W(M)\}$ & $\underset{M \in \mathcal{M}}{\min} \max \{d_U(M), d_W(M)\}$ \\
\\[-0.5em]
& Balanced stable matching \cite{Fed90} & Minimum regret stable matching \cite{Knu76} \\
\\
\multirow{2}{2.4cm}{Minimising the absolute difference} & $\underset{M \in \mathcal{M}}{\min} |c_U(M) - c_W(M)|$ & $\underset{M \in \mathcal{M}}{\min} |d_U(M) - d_W(M)|$ \\
\\[-0.5em]
& Sex-equal stable matching \cite{GI89} & Regret-equal stable matching *\\
\\
\multirow{2}{2.4cm}{Minimising the sum} & $\underset{M \in \mathcal{M}}{\min} (c_U(M) + c_W(M))$ & $\underset{M \in \mathcal{M}}{\min} (d_U(M) + d_W(M))$  \\
\\[-0.5em]
& Egalitarian stable matching \cite{Knu76} & Min-regret sum stable matching *\\
 \hline\hline \end{tabular}} \caption{Commonly used definitions of fair stable matchings in \acrshort{smi}. Our contributions are labelled with an *.} \label{group_sm_optimisation_options} \end{table} 
 
 In Table \ref{group_sm_optimisation_options} there are two new natural definitions of fairness that can be studied.

\begin{itemize}
	\item 
% 	A stable matching $M$ is \emph{regret-equal} if the absolute difference between $d_U(M)$ and $d_W(M)$ is minimum over all stable matchings in $\mathcal{M}$. 
	We define the \emph{regret-equality score} $r(M)$ as $|d_U(M) - d_W(M)|$ for a given stable matching $M$. $M$ is \emph{regret-equal} if $r(M)$ is minimum, taken over all stable matchings in $\cal M$. Note that in general we will prefer a regret-equal stable matching $M$ such that $d_U(M) + d_W(M)$ is minimised (e.g. $d(M) = (3, 3)$ rather than $d(M) = (10, 10)$). 
	\item 
% 	A stable matching $M$ is \emph{min-regret sum} if the sum of $d_U(M)$ and $d_W(M)$ is minimum over all stable matchings in $\mathcal{M}$. 
	We define the \emph{regret sum} as $d_U(M) + d_W(M)$ for a given stable matching $M$. $M$ is \emph{min-regret sum} if $d_U(M) + d_W(M)$ is minimum taken over all stable matchings in $\cal M$.
\end{itemize}

\subsection{Motivation}
 \label{sm_re_motivation}
Matching algorithms are widely used in the real world to solve allocation problems based on  \acrshort{smi} and its variants. A famous example of this is the \emph{\acrlong{nrmp}} (\acrshort{nrmp}). This scheme has been running in the US since 1952, and involves the allocation of thousands of graduating medical students to hospitals \cite{PR95}. Other matching schemes involve the allocation of students to projects \cite{AIM07} and the allocation of kidney donors to kidney patients \cite{BKMPABCDDHHJKKNSSVV19}.

Let mentees take the place of men and mentors take the place of women. Thus, mentees (mentors) rank a subset of mentors (mentees) and may only be allocated one mentor (mentee) in any matching. If we used the (renamed) Mentee-Oriented Gale-Shapley Algorithm \cite{GS62} to find a stable matching of mentees to mentors, then we would find a mentee-optimal stable matching $M$. However, as previously discussed, this would also be a mentor-pessimal stable matching. A similar but reversed situation happens using the (also renamed) Mentor-Oriented Gale-Shapley Algorithm \cite{GS62}. Therefore we may wish to find a stable matching that is in some sense fair between mentees and mentors using some of the criteria described in the previous section. All the types of fair stable matchings described in Table \ref{group_sm_optimisation_options} are viable candidates. However, as previously described, each of the problems of finding a balanced stable matching or a sex-equal stable matching is NP-hard, and there are existing polynomial time algorithms in the literature to find only two types of fair stable matchings, namely an egalitarian stable matching (in $O(m^{1.5})$ time) \cite{Fed90} and a minimum regret stable matching (in $O(m)$ time) \cite{Gus87}. Therefore, additional definitions of new, fair stable matchings and polynomial-time algorithms to calculate them provide additional choice for a matching scheme administrator. 

Moreover, we may be interested in finding a measure that gives a worst-off mentee a partner of rank as close as possible to that of a worst-off mentor. However, from our experimental work in Section \ref{sm_re_exps_sec}, we found that there was no other type of optimal stable matching that closely approximates the regret-equality score of the regret-equal stable matching. Indeed, results show that there exist regret-equal stable matchings with balanced score, cost and degree that are close to that of a balanced stable matching, an egalitarian stable matching and a minimum regret stable matching, respectively. This motivates the search for efficient algorithms to produce a regret-equal stable matching that has ``good'' measure relative to other types of fair stable matching.

 Whilst the practical motivation for studying min-regret sum stable matchings may not be as strong as in the regret-equality case, theoretical motivation comes from completing the study of the algorithmic complexity of computing all types of fair stable matchings relative to cost and degree, as shown in Table \ref{group_sm_optimisation_options}.

 \subsection{Contribution}
 \label{sm_re_contribution}

In this paper, we present two efficient algorithms: one to find a regret-equal stable matching, and one to find a min-regret sum stable matching, in an instance $I$ of \acrshort{smi}. Let $M_0$ and $M_z$ be the man-optimal and woman-optimal stable matchings in $I$. First we present the \textit{Regret-Equal Degree Iteration Algorithm} (REDI), to find a regret-equal stable matching in an instance $I$ of \acrshort{smi}, with time complexity $O(d_0 nm)$, where $d_0 = |d_U(M_0) - d_W(M_0)|$. This is the main result of the paper. Second we present the \textit{Min-Regret Sum Algorithm} (MRS), to find a min-regret sum stable matching in an instance $I$ of \acrshort{smi}, with time complexity $O(d_s m)$,  where $d_s = d_U(M_z) - d_U(M_0)$. In addition to this theoretical work, the REDI algorithm was implemented and its performance was compared against an algorithm to enumerate all stable matchings \cite{Gus87} (exponential in the worst case). Finally, experiments were conducted to compare six different types of optimal stable matchings (balanced, sex-equal, egalitarian, min-regret, regret-equal, min-regret sum), and output from Algorithm REDI, over a range of measures (including balanced score, sex-equal score, cost, degree, regret-equality score, regret sum). In addition to the observations already discussed in Section \ref{sm_re_motivation}, we found a large variation in sex-equal scores and regret-equality scores among the six different types of optimal stable matching, and, a far smaller variation for the balanced score, cost, degree and regret sum measures. This smaller variation also includes outputs of Algorithm REDI, indicating that we are able to find a regret-equal stable matching in polynomial time with a likely good balanced score, cost and degree using this algorithm. Indeed, we find in practice that Algorithm REDI approximates these types of optimal stable matchings at an average of $9.0\%$, $1.1\%$ and $3.0\%$ over their respective optimal values, for randomly-generated instances with $n=1000$.

 \subsection{Structure of the paper}
 \noindent Section \ref{sm_re_fdotp} describes a \emph{rotation} and related concepts in \acrshort{smi} that will be used later in the paper. Sections \ref{sm_re_di_sec} and \ref{sm_mrs_sec} describe Algorithm REDI and Algorithm MRS respectively, giving in each case pseudocode, correctness proofs and time complexity calculations. An experimental evaluation is given in Section \ref{sm_re_exps_sec}. Finally, future work is presented in Section \ref{sm_re_fw_sec}.

\section{Structure of stable matchings}
\label{sm_re_fdotp}

For some stable matching $M$ in an instance $I$ of \acrshort{smi}, let $s(m_i, M)$ denote the next woman on $m_i$'s preference list (starting from $M(m_i)$) who prefers $m_i$ to $M(s(m_i, M))$ (their partner in $M$). A \emph{rotation} $\rho$ is then a sequence of man-woman pairs $\{(m_1, w_1), (m_2, w_2), ..., (m_q, w_q)\}$ in $M$, such that $m_{i+1} = M(s(m_i, M))$ for $1 \leq i \leq q$ where addition is taken modulo $q$ \cite{ILG87}. We say rotation $\rho$ is \emph{exposed} on $M$ if $\{(m_1, w_1), (m_2, w_2), ..., (m_q, w_q)\} \subseteq M$. If $\rho$ is exposed on $M$, we may \emph{eliminate} $\rho$ on $M$, that is, remove all pairs of $\rho$ from $M$ and add pairs $(m_i, w_{i+1})$ for $1 \leq i \leq q$, where addition is taken modulo $q$, in order to produce another stable matching $M'$ of $I$. The \emph{rotation poset} $R_p(I)$ of $I$ indicates the order in which rotations may be eliminated. Rotation $\rho$ is said to \emph{precede} rotation $\tau$ if $\tau$ is not exposed until $\rho$ has been eliminated. There is a one-to-one correspondence between the set of stable matchings and the set of closed subsets of $R_p(I)$ \cite[Theorem 3.1]{ILG87}. Gusfield and Irving \cite{GI89} describe a graphical structure known as the \emph{rotation digraph} $R_d(I)$ of $I$ which is based on $R_p(I)$ and allows for the enumeration of all stable matchings in $O(m+n|\mathcal{M}|)$ time. 

Let $R$ be the set of rotations of $I$. Then $R_j(M)$ is the set of rotations that contain a women of rank $j$ in $M$, that is, $R_j(M)=\{\sigma \in R : (m, w)\in \sigma \wedge \text{rank}(w, M(w))=j) \}$. Let $M_z$ be the woman-optimal stable matching \cite{GS62}. For any stable pair $(m_i, w_j) \notin M_z$, let $\phi(m_i, w_j)$ denote the unique rotation containing pair $(m_i, w_j)$. Finally, denote by $c(\rho)$ the closure of rotation $\rho$ and similarly denote by $c(R')$  the closure of set of rotations $R'$. We say that the closure of an undefined rotation or an empty set of rotations is the empty set. 

\section{Algorithm to find a regret-equal stable matching in SMI}
\label{sm_re_di_sec}

\subsection{Description of the Algorithm}
\label{sm_re_doa}

Algorithm REDI, which finds a regret-equal stable matching in a given instance $I$ of \acrshort{smi}, is presented as Algorithm \ref{alg_sm_reg_eq}. For an instance $I$ of \acrshort{smi}, Algorithm REDI begins with operations to find the man-optimal and woman-optimal stable matchings, $M_0$ and $M_z$, found using the Man-oriented and Women-oriented Gale-Shapley Algorithm \cite{GS62}. The set of rotations $R$ is also found using the Minimal Differences Algorithm \cite{ILG87}. 

Let $d(M_0)=(a_0,b_0)$. If $a_0=b_0$ then we must have an optimal stable matching and so we output $M_0$ on Line \ref{alg_sm_reg_eq_M0_return}. If $a_0 > b_0$ then any other matching $M'$, where $d(M') = (a', b')$, must have $a' \geq a_0$ and $b' \leq b_0$ since any rotation (or combination of rotations) eliminated on the man-optimal matching $M_0$ will make men no better off and women no worse off. Therefore $M_0$ is optimal and so it is returned on Line \ref{alg_sm_reg_eq_M0_return}. Now suppose $a_0 < b_0$. Throughout the algorithm we save the best matching found so far to the variable $M_{opt}$ starting with $M_0$. We know that a matching exists with $d_0=b_0 - a_0$ and so we try to improve on this, by finding a matching $M$ with $r(M) < d_0$. 

We create several `columns' of possible degree pairs of a regret-equal matching as follows. The top-most pairs for columns $k \geq 1$ are given by the sequence 
$$\big( (a_0, b_0), (a_0 + 1, b_0), (a_0 + 2, b_0), ..., (\min \{n, 2b_0 - a_0 - 1\}, b_0) \big).$$

The sequence of pairs for column $k$ $(1 \leq k \leq \min \{2d_0, n-a_0+1\})$ from top to bottom is given by
$$\big( (a_0 + k - 1, b_0), (a_0 + k - 1, b_0 - 1), (a_0 + k - 1, b_0 - 2), ..., (a_0 + k - 1, \max \{a_0 - d_0 + k, 1\}) \big).$$

At this point as long as the size $n$ of the instance satisfies $n \geq 2b_0 - a_0 - 1$ and $a_0 - d_0 + 1 \geq 1$, the possible degree pairs of a regret-equal matching are shown in Figure \ref{re_ex_differences} of Appendix \ref{sm_re_degpairs_supp}. We know this accounts for all possible degree pairs since, as above, if $M'$ is any  matching not equal to $M_0$, where $d(M') = (a', b')$, it must be that $a' \geq a_0$ and $b' \leq b_0$. Setting $b' = b_0$, the largest $a'$ could be is given by $b_0$ added to the maximum possible improved difference $d_0 - 1$, that is, $a'=b_0+d_0-1 = 2b_0 - a_0 - 1$. If $n < 2b_0 - a_0 - 1$ then we only consider the first $n - a_0 + 1$ columns in Figure \ref{re_ex_differences}. The $a_0 - d_0 + k$ value is obtained by noting that if $x$ is the final value of women's degree for the column sequence above then $a_0 + k - 1 - x = d_0 - 1$ and so $x = a_0 + k - d_0$. 
Figure \ref{re_ex_differences_2} of Appendix \ref{sm_re_degpairs_supp} shows an example of the possible regret-equal degree pairs when $d(M_0) = (2, 6)$ and $n \geq 9$.

The column operation (Algorithm \ref{alg_sm_reg_eq_sub}) works as follows. Let local variable $M$ hold the current matching for this column, and let local variable $Q$ be the set of rotations corresponding to $M$. Iteratively we first test if $r(M) < r(M_{opt})$ setting $M_{opt}$ to $M$ if so. We now check whether $d_U(M) \geq d_W(M)$. If it is, then any further rotation for this column will only make $r(M)$ larger, and so we stop iterating for this column, returning $M_{opt}$. Next, we find the set of rotations $Q'$ in the closure of $R_b(M) \subseteq R$ that are not already eliminated to reach $M$. If eliminating these rotations would either increase the men's degree or not decrease the women's degree, then we return $M_{opt}$. Otherwise, set $M$ to be the matching found when eliminating these rotations. 

If after the column operation, $d_U(M_{opt}) = d_W(M_{opt})$, then we have a regret-equal matching and it is immediately returned on Lines \ref{alg_sm_reg_eq_first_return} or \ref{alg_sm_reg_eq_second_return} of Algorithm \ref{alg_sm_reg_eq}.

%%%%%%%%%%%%%%%%%%%%%%%%%%%%%%%%%%%%%%%%%%%%%%%%%%%%%%%%%
%%%%%%%%%%%%%%%%%%%%%%%%%%%%%%%%%%%%%%%%%%%%%%%%%%%%%%%%%
% ALGORITHM 

\begin{algorithm} [t!]

  \caption{REDI($I$), returns a regret-equal stable matching for an instance $I$ of \acrshort{smi}.}
	\begin{algorithmic}[1]
          \Require An instance $I$ of \acrshort{smi}.
          \Ensure Return a regret-equal stable matching $M_{opt}$.           
%          \State Let $R(I)$ denote the set of rotations in the man-oriented rotation poset.
          \State $M_0 \gets$ MGS($I$) \Comment $M_0$ is the man-optimal stable matching found using the Man-oriented Gale-Shapley Algorithm (MGS) \cite{GS62}. \label{alg_sm_reg_eq_man-opt}
          \State $M_z \gets$ WGS($I$) \Comment $M_z$ is the woman-optimal stable matching found using the Woman-oriented Gale-Shapley Algorithm (WGS) \cite{GS62}. 

          \State $R \gets$ MIN-DIFF($I$) \Comment $R$ is the set of rotations found using the Minimal Differences Algorithm (MIN-DIFF) \cite{ILG87}.
          \If {$d_U(M_0) \geq d_W(M_0)$} 
          \State \Return $M_0$ \label{alg_sm_reg_eq_M0_return}
          \EndIf
          \State $M_{opt} \gets M_0$ \Comment $M_{opt}$ is the best stable matching found so far.
%          \State $d \gets d_W(M_0) - d_U(M_0)$ \Comment $d$ is the smallest difference in degree so far.
          	\State $M_{opt} \gets $REDI-COL$(I, M_0, \emptyset, M_{opt})$ \Comment Find the best matching for the first column. \label{alg_sm_reg_eq_subcall_1}
          	\If {$r(M_{opt}) = 0$}
          \State \Return $M_{opt}$ \label{alg_sm_reg_eq_first_return}
          \EndIf

                    \ForAll{$m_i \in U$} \Comment For each man. \label{alg_sm_reg_eq_for}
          \State $M \gets M_0$ \Comment $M$ is the matching we start from for $m_i$ at the beginning of each column.
          	\State $Q \gets \emptyset$ \Comment $Q$ is the set of rotations corresponding to $M$.
          \While {$(m_i, M(m_i)) \notin M_z$ \textbf{and} $d_U(M) < d_W(M)$} \label{alg_sm_reg_eq_while}
          \State $\rho = \phi(m_i, M(m_i))$
          \State $a \gets d_U(M)$
          \State $Q' \gets c(\rho) \backslash Q$ 
          \State $M \gets M / Q'$ \Comment Rotations in $Q'$ are eliminated in order defined by the rotation poset of $I$.\label{alg_sm_reg_eq_elimination_1}
          	\State $Q \gets Q \cup Q'$
%          \State $M \gets $ REG-EQ-NEXT$(m_i, M, \rho)$ \Comment Return the matching found when rotating $m_i$ down once.
		\If {$d_U(M) > a$ \textbf{and} $\text{rank}(m_i, M(m_i))=d_U(M)$} \Comment The men's degree has increased and $m_i$ is a worst ranked man in $M$. \label{sm_re_alg_condition_for_colop}
          \State $M_{opt} \gets$REDI-COL$(I, M, Q, M_{opt})$ \Comment Find the best matching for this column. \label{alg_sm_reg_eq_subcall_2}
          	\If {$r(M_{opt}) = 0$}
          \State \Return $M_{opt}$ \label{alg_sm_reg_eq_second_return}
          \EndIf
         	\EndIf
          \EndWhile
          \EndFor 
          \State \Return $M_{opt}$ \label{alg_sm_reg_eq_third_return}

     	\end{algorithmic}
	\label{alg_sm_reg_eq}
\end{algorithm}
% ALGORITHM END

The column operation described above is called first from the man-optimal stable matching $M_0$ on Line \ref{alg_sm_reg_eq_subcall_1}, to iterate down the first column. Then for each man $m_i$ we do the following. Let $M$ be set to $M_0$. Iteratively we eliminate $(m_i, M(m_i))$ from $M$ by eliminating rotation $\rho$ and its predecessors (not already eliminated to reach $M$) such that $(m_i, M(m_i)) \in \rho$. We continue doing this until both the men's degree increases and $\text{rank}(m_i,M(m_i))=d_U(M)$ (in the same operation). This has the effect of jumping our focus from some column of possible degree pairs, to another column further to the right with $m_i$ being one of the lowest ranked men in $M$. Once we have moved to a new column we perform the column operation described above. If either $m_i$ has the same partner in $M$ as in $M_z$ (hence there are no rotations left that move $m_i$) or   $d_U(M) > d_W(M)$ (further rotations will only increase the regret-equality score), then we stop iterating for $m_i$. In this case we restart this process for the next man, or return $M_{opt}$ if we have completed this process for all men. Note that since at the end of a while loop iteration, if $r(M) = 0$ then $M_{opt}$ is returned, it is not possible for the condition $d_U(M) = d_W(M)$ to ever be satisfied in the while loop clause.

% ALGORITHM 
\begin{algorithm} [t!]
  \caption{REDI-COL$(I, M, Q, M_{opt})$, subroutine for Algorithm \ref{alg_sm_reg_eq}. Column operation for the current column $d_U(M)$. Returns $M_{opt}$, the best stable matching found so far (according to the regret-equality score).}
	\begin{algorithmic}[1]
          \Require An instance $I$ of \acrshort{smi}, stable matching $M$, the closure of $M$, $Q$ and $M_{opt}$ the best stable matching found so far (according to the regret-equality score).
          \Ensure Finds the best stable matching (according to the regret-equality score) found when incrementally eliminating women of worst rank from the current matching, without increasing the men's degree. If an improvement is made then $M_{opt}$ is updated. $M_{opt}$ is returned. All variables used within Algorithm \ref{alg_sm_reg_eq_sub} are understood to be local.
          
          \State $a \gets d_U(M)$
%          	\State $Q'' \gets Q$ \Comment $Q''$ is the set of rotations corresponding to $M$.
		\While {\textbf{true}} \label{alg_sm_reg_eq_sub_while}
		\If {$r(M) < r(M_{opt})$} 
          \State $M_{opt} \gets M$ \label{alg_sm_reg_eq_sub_setToMopt}
          \EndIf
          \If {$d_U(M) \geq d_W(M)$} \Comment Further rotations for this column would only increase the difference in degree of men and women.
          \State \Return $M_{opt}$
          \EndIf

			\State $b \gets d_W(M)$
			\State $Q' \gets c(R_b(M)) \backslash Q$ 
          	\If {$d_U(M/Q')>a \lor d_W(M/Q')=b$}
          	\State \Return $M_{opt}$
          	\Else
          \State $M \gets M / Q'$ \Comment Rotations in $Q'$ are eliminated in order defined by the rotation poset of $I$.\label{alg_sm_reg_eq_sub_elimination_2}
          	\State $Q \gets Q \cup Q'$
			          \EndIf
          \EndWhile
	\end{algorithmic}
	\label{alg_sm_reg_eq_sub}
\end{algorithm}

% ALGORITHM END

\subsection{Correctness proof and time complexity}
In this section we state the correctness and time complexity results for Algorithm REDI. The proofs of these theorems may be found in Appendix \ref{app_redi_corr_proof}.

\begin{restatable*}{theorem}{smreappalgredithm}
	Let $I$ be an instance of \acrshort{smi}. Any matching produced by Algorithm REDI is a regret-equal stable matching of $I$. 
	\label{sm_re_app_alg_redi_thm}
\end{restatable*}

\begin{restatable*}{theorem}{smreappalgreditime}
	Let $I$ be an instance of \acrshort{smi}. Algorithm REDI always terminates within $O(d_0nm)$ time, where $d_0=|d_U(M_0) - d_W(M_0)|$, $n$ is the number of men or women in $I$, $m$ is the total length of all preference lists and $M_0$ is the man-optimal stable matching.
	\label{sm_re_app_alg_redi_time}
\end{restatable*}

\subsection{Regret-equal stable matchings with minimum cost}

We may seek a regret-equal stable matching with minimum cost over all regret-equal stable matchings. This may be achieved in $O(nm^{2.5})$ time using the following process. 

We define the \emph{deletion} of pair $(m_i, w_j)$ as the removal of $w_j$ from $m_i$'s preference list and the removal of $m_i$ from $w_j$'s preference list. \emph{Truncating men's preference lists at $t$}, where $1 \leq t \leq n$, is then the process of deleting  pair $(m_i, w_j)$ for each $(m_i, w_j)$ such that $\text{rank}(m_i, w_j) > t$. An analogous definition holds for women. For a given SMI instance $I$, first find the regret-equality score $r$ of the regret-equal stable matching using Algorithm REDI in $O(d_0nm)$ time. Then, iterate over all possible man-woman degree pairs $(a, b)$ such that $|a-b|=r$ (there are $O(n)$ such pairs). For each such degree pair $(a,b)$, truncate men at $a$ and women at $b$, creating instance $I'$. Then, for each of the $O(m)$ man-woman pairs $(m_i, w_j)$ in $I'$, fix $m_i$ with his $a$th-choice partner and $w_j$ with her $b$th-choice partner (where ranks are taken with respect to instance $I$), if possible. If this is not possible then continue to the next degree pair. Assume that $w_j'$ is $m_i$'s $a$th-choice partner, and $m_i'$ is $w_j$'s $b$th-choice partner. In $I'$, we now delete pairs $(m_i'',w_j'')$ for any $w_j''$ such that $m_i$ prefers $w_j''$ to $w_j'$ and $w_j''$ prefers $m_i$ to $m_i''$.  Also delete the pair  $(m_i'',w_j'')$ for any $m_i''$ such that $w_j$ prefers $m_i''$ to $m_i'$ and $m_i''$ prefers $w_j$ to $w_j''$. Next we delete all remaining preference list elements of $m_i$ except $w_j'$ and all remaining preference list elements of $w_j$ except $m_i'$. The Gale-Shapley Algorithm is run to check that a stable matching of size $n$ exists in $I'$. If no such stable matching exists then we move on to the next degree pair. Feder's Algorithm may then be used to find an egalitarian stable matching in the reduced \acrshort{smi} instance $I'$ in $O(m^{1.5})$ time (using the original ranks in $I$ as costs). This makes a total of $O(nm^{2.5})$ time to find a regret-equal stable matching with minimum cost.

\section{Algorithm to find a min-regret sum stable matching in SMI}
\label{sm_mrs_sec}

Algorithm MRS, which finds a min-regret sum stable matching, given an instance of \acrshort{smi}, is presented as Algorithm \ref{sm_mrs_alg}. First, the man-optimal and woman-optimal stable matchings, $M_0$ and $M_z$, are found using the Man-oriented and Women-oriented Gale-Shapley Algorithms \cite{GS62}. The best matching found so far, denoted $M_{opt}$ is initialised to $M_0$. We then iterate over each possible man degree $a$ between $d_U(M_0)$ and $d_U(M_z)$ inclusive, where an improvement of $M_{opt}$, according to the regret sum, is still possible. As an example, suppose $M_{opt}$ has a regret sum of $5$ with $d_U(M_{opt}) = 2$ and $d_W(M_{opt}) = 3$. Then, it is not worth iterating over any man degree greater than $3$ since it will not be possible to improve on the regret sum of $5$ by doing so. For each iteration of the while loop, we truncate the men's preference lists at $a$, and find the woman-optimal stable matching $M_z^T$ for this truncated instance. If the regret sum of $M_z^T$ is smaller than that of $M_{opt}$, we update $M_{opt}$ to $M_z^T$. After all iterations over possible men's degrees are completed, $M_{opt}$ is returned. 

\begin{algorithm} [ht]
  \caption{MRS$(I)$, returns a min-regret sum stable matching for an instance $I$ of \acrshort{smi}}
	\begin{algorithmic}[1]
	\Require An instance $I$ of \acrshort{smi}.
	\Ensure Return a min-regret sum stable matching $M_{opt}$.
	\State $M_0 \gets$ MGS($I$) \Comment $M_0$ is the man-optimal stable matching found using the Man-oriented Gale-Shapley Algorithm (MGS) \cite{GS62}. 
	\State $M_z \gets$ WGS($I$) \Comment $M_z$ is the woman-optimal stable matching found using the Woman-oriented Gale-Shapley Algorithm (WGS) \cite{GS62}. 
    \State $M_{opt} \gets M_0$
    
    \State $a \gets d_U(M_0)$
    \While {$a \leq d_U(M_z)$ \textbf{and} $a + 1 < d_U(M_{opt}) + d_W(M_{opt})$} \label{sm_mrs_alg_while}
    \State $I_T \gets $instance $I$ where men's preference lists are truncated at rank $a$. \label{sm_msr_alg_trunc}
    \State $M_z^T \gets$ WGS($I_T$) \label{sm_msr_alg_trunc_wopt}
    \If {$d_U(M_z^T) + d_W(M_z^T) < d_U(M_{opt}) + d_W(M_{opt})$}
    \State $M_{opt} \gets M_z^T$ \label{sm_mrs_alg_mopt_updated}
    \EndIf
    \State $a \gets a + 1$
    \EndWhile
    \State \Return $M_{opt}$ \label{sm_msr_alg_return_mopt}
	\end{algorithmic}
	\label{sm_mrs_alg}
\end{algorithm}
% ALGORITHM END

Let $d_s$ denote the difference between the degree of men in the woman-optimal stable matching $M_z$, and in the man-optimal stable matching $M_0$, that is $d_s = d_U(M_z) - d_U(M_0)$. Theorem \ref{sm_msr_alg_proof} as follows states that Algorithm MRS produces a min-regret sum stable matching in $O(d_s m)$ time. See Appendix \ref{app_mrs_corr_tc_proof} for the proof of this Theorem.

\begin{restatable*}{theorem}{smmsralgproof}
Let $I$ be an instance of \acrshort{smi}. Algorithm MRS produces a min-regret sum stable matching in $O(d_s m)$ time, where $d_s = d_U(M_z) - d_U(M_0)$, $m$ is the total length of all preference lists, and $M_0$ and $M_z$ are the man-optimal and woman-optimal stable matchings respectively.
	\label{sm_msr_alg_proof}
\end{restatable*}

\section{Experiments}
\label{sm_re_exps_sec}

\subsection{Methodology}

An Enumeration Algorithm (ENUM) exists to find the set of all stable matchings of an instance $I$ of \acrshort{smi} in $O(m + n|\mathcal{M}|)$ time \cite{Gus87}. Within this time complexity, it is possible to output a regret-equal stable matching from this set of stable matchings, by keeping track of the best stable matching found so far (according to the regret-equality score) as they are created. We randomly generated instances of \acrshort{sm}, in order to compare the performance of Algorithms REDI and ENUM. Using output from Algorithm ENUM, we also investigated the effect of varying instance sizes, for six different types of optimal stable matchings (balanced, sex-equal, egalitarian, min-regret, regret-equal, min-regret sum), and also output from Algorithm REDI, over a range of measures (including balanced score, sex-equal score, cost, degree, regret-equality score, regret sum). Tests were run over $19$ different instance types with varying instance size ($n \in \{10, 20, ..., 100, 200, ..., 1000\}$). All instances tested were complete with uniform distributions on preference lists. Experiments were run over $500$ instances of each instance type.

Each instance was run over the two algorithms described above with a timeout time of $1$ hour for each algorithm. No instances timed out for these experiments. Experiments were conducted on a machine running Ubuntu version $18.04$ with 32 cores, 8$\times$64GB RAM and Dual Intel\textsuperscript{\textregistered} Xeon\textsuperscript{\textregistered} CPU E5-2697A v4 processors. Instance generation, correctness and statistics summarisation programs, and plot and \LaTeX\ table generation were all written in Python and run on Python version $3.6.1$. All other code was written in Java and compiled using Java version $1.8.0$. Each instance was run on a single thread with $16$ instances run in parallel using GNU Parallel \cite{Tan11}. Serial Java garbage collection was used with a maximum heap size of $2$GB distributed to each thread. Code and data repositories for these experiments can be found at \href{https://zenodo.org/record/3630383}{zenodo.org/record/3630383} and \href{https://zenodo.org/record/3630349}{zenodo.org/record/3630349} respectively. Comprehensive correctness testing was conducted, a description of which may be seen in Appendix \ref{sm_re_exps_sec_correctness_app}.

\subsection{Experimental results summary}
Figure \ref{reg-eq-duration} shows a comparison of the time taken to execute the two algorithms over increasing values of $n$. Precise data for this plot can be seen in Table \ref{sm_re_table_duration} of Appendix \ref{sm_re_exps_sec_f_t_app} which gives the mean, median, $5$th percentile and $95$th percentile durations for Algorithms REDI and ENUM. In Figure \ref{reg-eq-duration}, the median values of time taken for each algorithm are plotted and a $90\%$ confidence interval is displayed using the $5$th and $95$th percentile measurements. Additional experiments and evaluations not discussed here may also be found in Appendix \ref{sm_re_exps_sec_additional_app}.

Figure \ref{reg-eq-opt-plots} shows comparisons of six different types of optimal stable matchings (balanced, sex-equal, egalitarian, min-regret, regret-equal, min-regret sum), and output from Algorithm REDI, over a range of measures (including balanced score, sex-equal score, cost, degree, regret-equality score, regret sum), as $n$ increases. Optimal stable matching statistics involving a measure determined by cost (respectively degree) are given a green (respectively blue) colour. For a particular fairness objective A and a particular fairness measure B, there may be a set of several stable matchings that are optimal with respect to A. In this case we choose a matching from this set that has best possible measure with respect to B. For example, if we are looking at the regret-equality score, for a particular instance, we find a sex-equal stable matching that has smallest regret-equality score (over the set of all sex-equal stable matchings) and use this value to plot the regret-equality score for this type of optimal stable matching. This process is replicated for the other types of optimal stable matching. In each case the mean measure value is plotted for the given type of optimal stable matching. Data for these plots may be found in Tables \ref{sm_re_table_balancedScoreAv}, \ref{sm_re_table_sexEqualCostAv},  \ref{sm_re_table_egalCostAv}, \ref{sm_re_table_degreeAv}, \ref{sm_re_table_regretEqualScoreAv} and \ref{sm_re_table_sumRegretAv} of Appendix \ref{sm_re_exps_sec_f_t_app}.

The main results of these experiments are:
\begin{itemize}
	\item \emph{Time taken:} It is clear from Figure \ref{reg-eq-duration} in  that Algorithm REDI is the faster algorithm in practice, taking approximately $2$s to solve an instance of size $n=1000$ with very little variation. In contrast, Algorithm ENUM takes around $8$s for an instance of size $n=1000$ with a far larger variation. 
	
	\item \emph{Sex-equal score:} A wide variation in sex-equal score over the six optimal matchings can be seen in Figure \ref{reg-eq-opt-plots_se} (and Table \ref{sm_re_table_sexEqualCostAv}). Sex-equal and balanced stable matchings are extremely closely aligned giving a mean sex-equal score of $265.0$ and $284.0$ respectively for the instance type with $n=1000$. Min-sum regret stable matchings, on the other hand, performed the least well with a mean sex-equal score of $12400.0$ for the same instance type. 
		
	\item \emph{Regret-equality score:} Similar to the previous point we see a wide variation in regret-equality score over the six optimal stable matchings in Figure \ref{reg-eq-opt-plots_re} (and Table \ref{sm_re_table_regretEqualScoreAv}). For the instance type with $n=1000$, this ranges from a mean regret-equality score of $14.2$ for the regret-equal stable matching to $84.6$ for the minimum regret stable matching. It is interesting to note that the type of optimal stable matching (out of the six optimal stable matchings tested) whose regret-equality score tends to be furthest away from that of a regret-equal stable matching is the min-regret sum stable matching. This may be due to the fact that minimising the sum of two measures does not necessarily force the two measures to be close together.
	
	\item \emph{Output from Algorithm REDI:} Due to the wide variation of regret-equality scores among different types of optimal stable matchings (as described above) it is clear that no other optimal stable matching is able to closely approximate a regret-equal stable matching, which highlights the importance of Algorithm REDI that is designed specifically for optimising this measure. Interestingly, Algorithm REDI is also competitive in terms of balanced score, cost and degree.  Indeed, we can see from Tables \ref{sm_re_table_balancedScoreAv}, \ref{sm_re_table_egalCostAv} and \ref{sm_re_table_degreeAv}, that Algorithm REDI approximates these types of optimal stable matchings at an average of $9.0\%$, $1.1\%$ and $3.0\%$ over their respective optimal values, for instances with $n=1000$. Over all instance sizes, these values are within ranges $[4.0\%, 10.9\%]$, $[1.1\%, 3.4\%]$ and $[1.3\%, 3.7\%]$, respectively. This gives a good indication of the high-quality of output from this algorithm even on seemingly unrelated measures.
\end{itemize}

\begin{figure}
	\centering
    \includegraphics[scale=0.52]{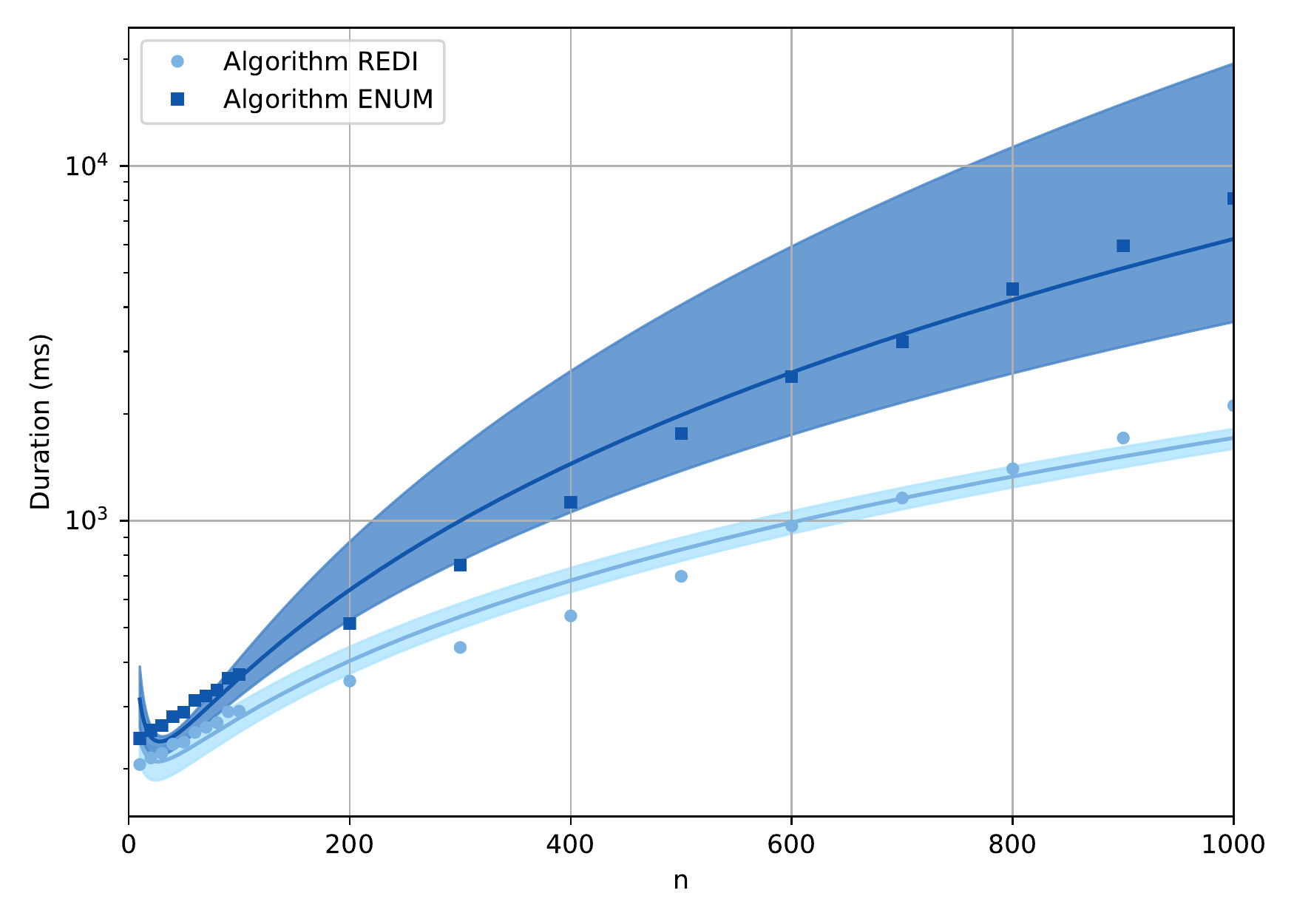}
    \caption{A $\log$ plot of the time taken to execute Algorithms REDI and ENUM. A second order polynomial model has been assumed for best-fit lines.}
    \label{reg-eq-duration}
\end{figure}

\begin{figure}[H]
\centering
\captionsetup[subfigure]{justification=centering}
  \begin{subfigure}[b]{0.48\textwidth}
  \centering
    \includegraphics[scale=0.36]{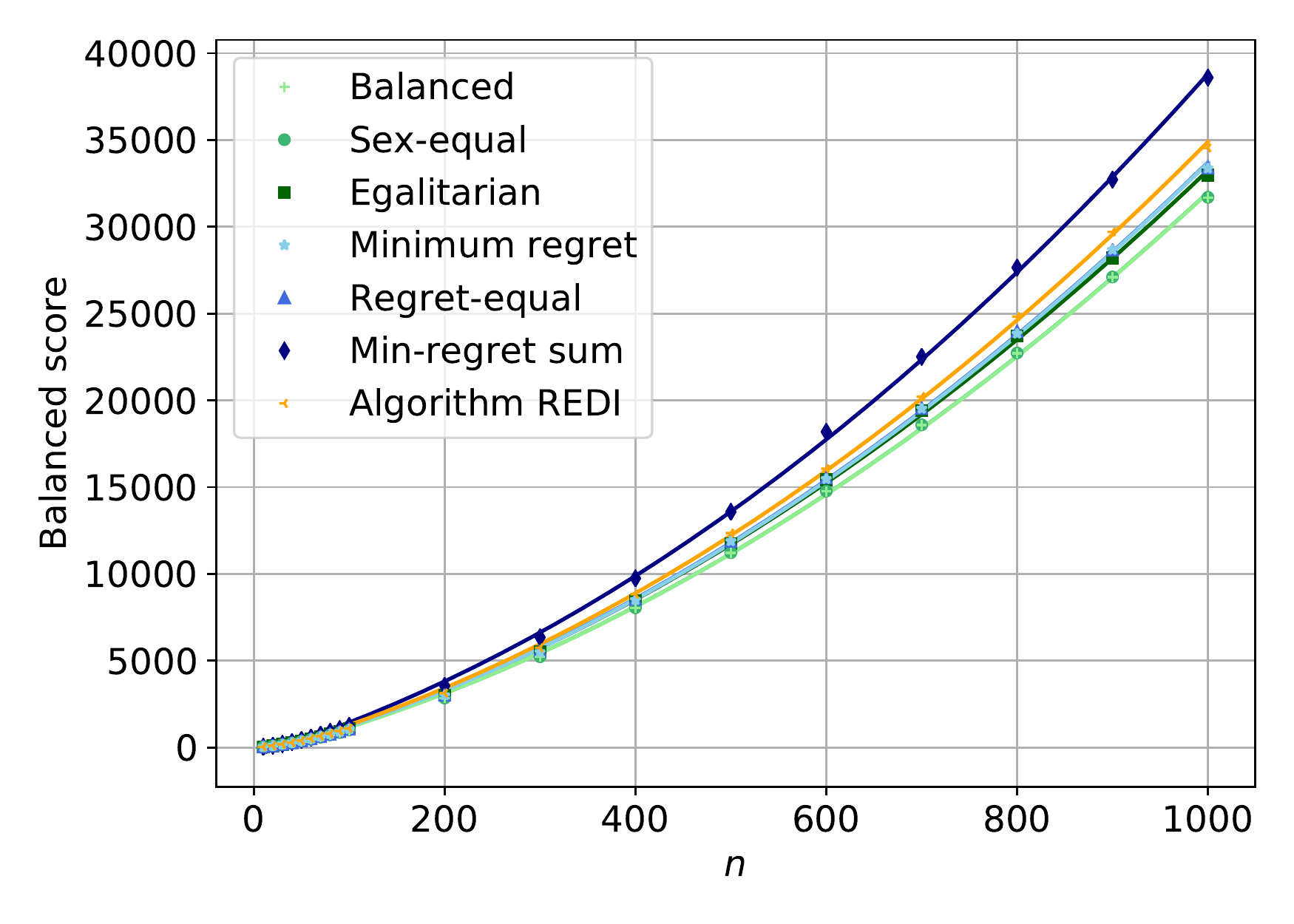}
    \caption{Plot of balanced score.}
    % \label{}
  \end{subfigure}
  \begin{subfigure}[b]{0.48\textwidth}
  \centering
    \includegraphics[scale=0.36]{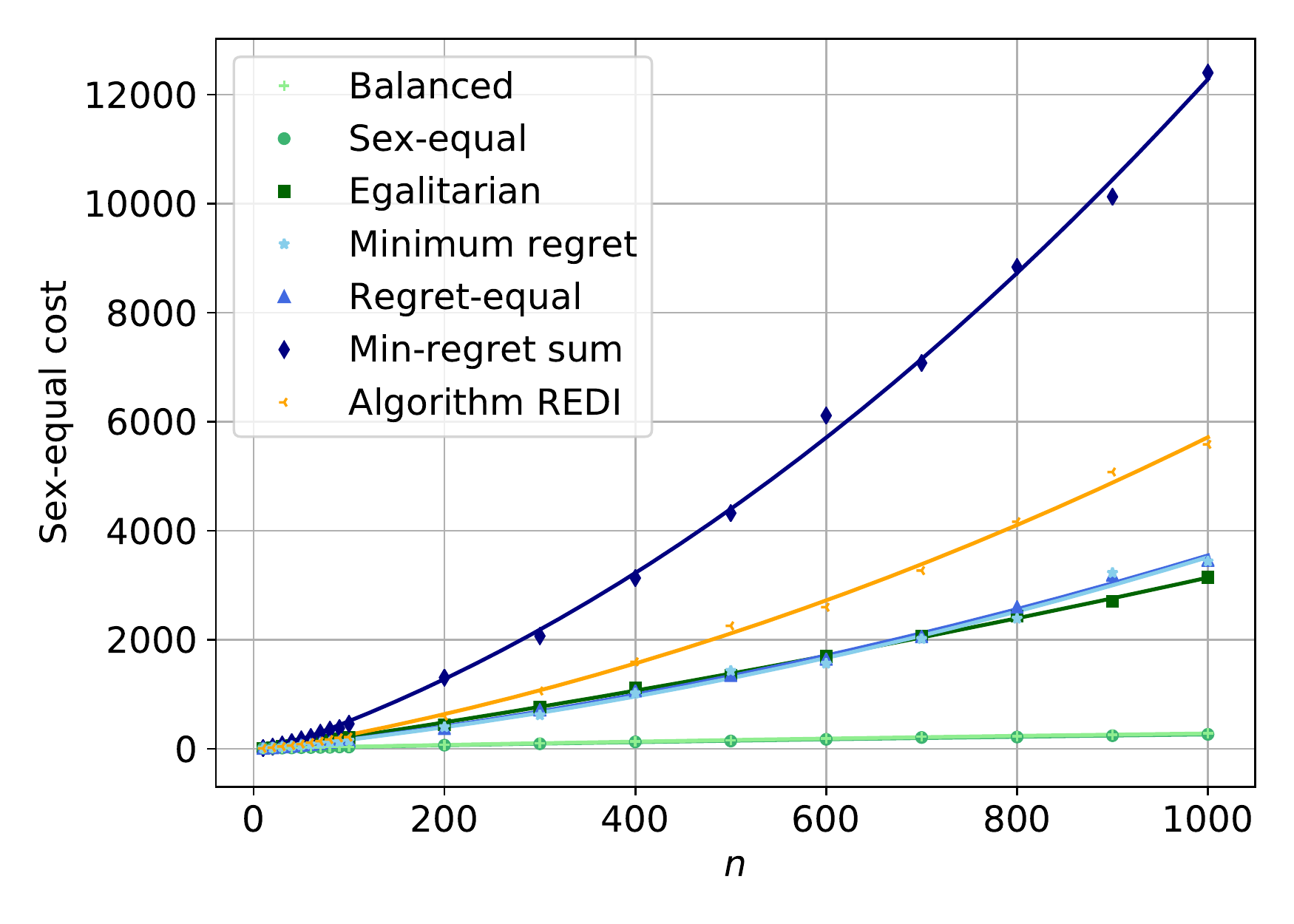}
    \caption{Plot of sex-equal score.}
    \label{reg-eq-opt-plots_se}
  \end{subfigure} 

    \begin{subfigure}[b]{0.48\textwidth}
  \centering
    \includegraphics[scale=0.36]{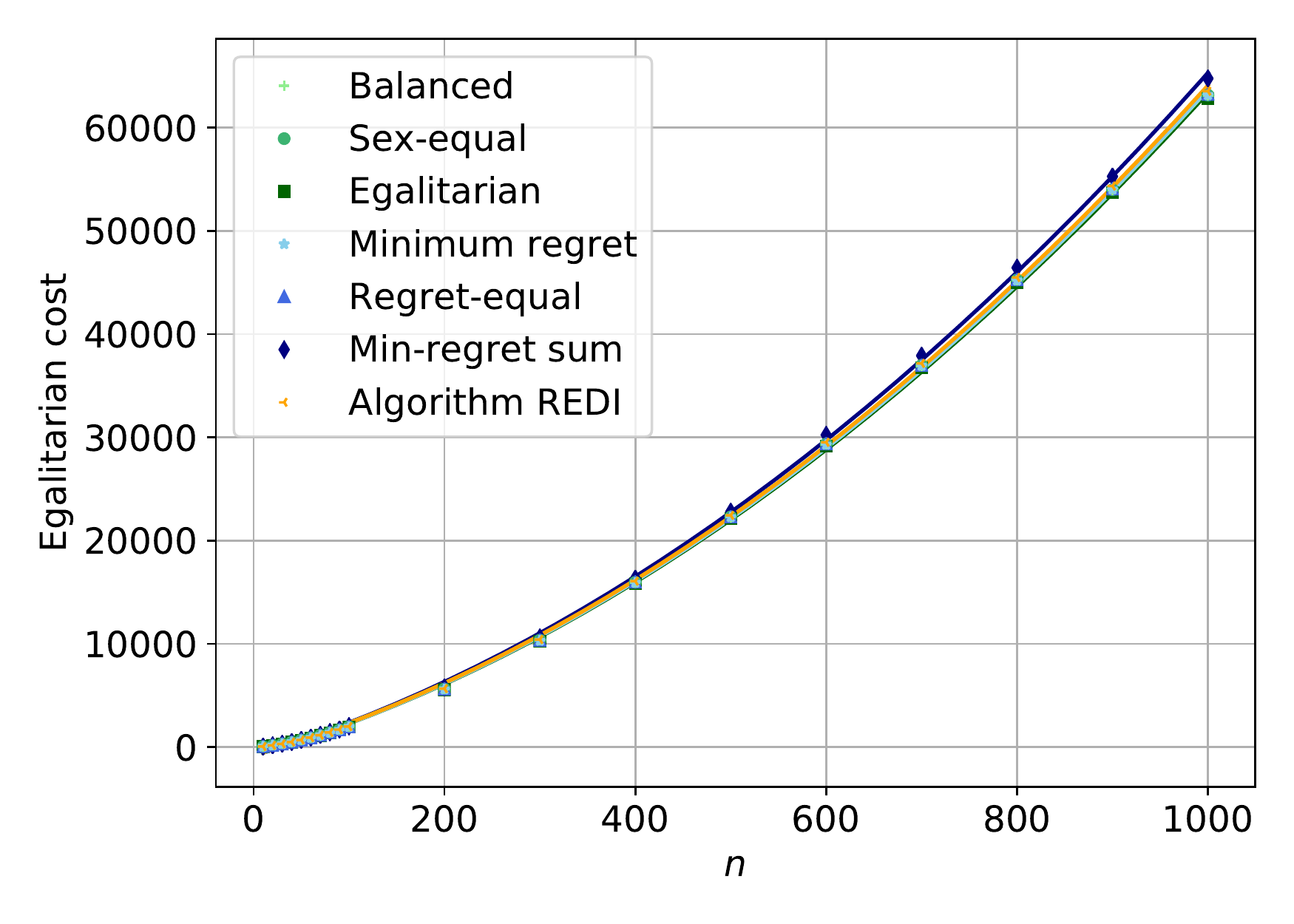}
    \caption{Plot of cost.}
    % \label{}
  \end{subfigure} 
% \end{figure}
% 
% \begin{figure}
% \ContinuedFloat
% \centering
% \captionsetup[subfigure]{justification=centering}
\begin{subfigure}[b]{0.48\textwidth}
  \centering
    \includegraphics[scale=0.36]{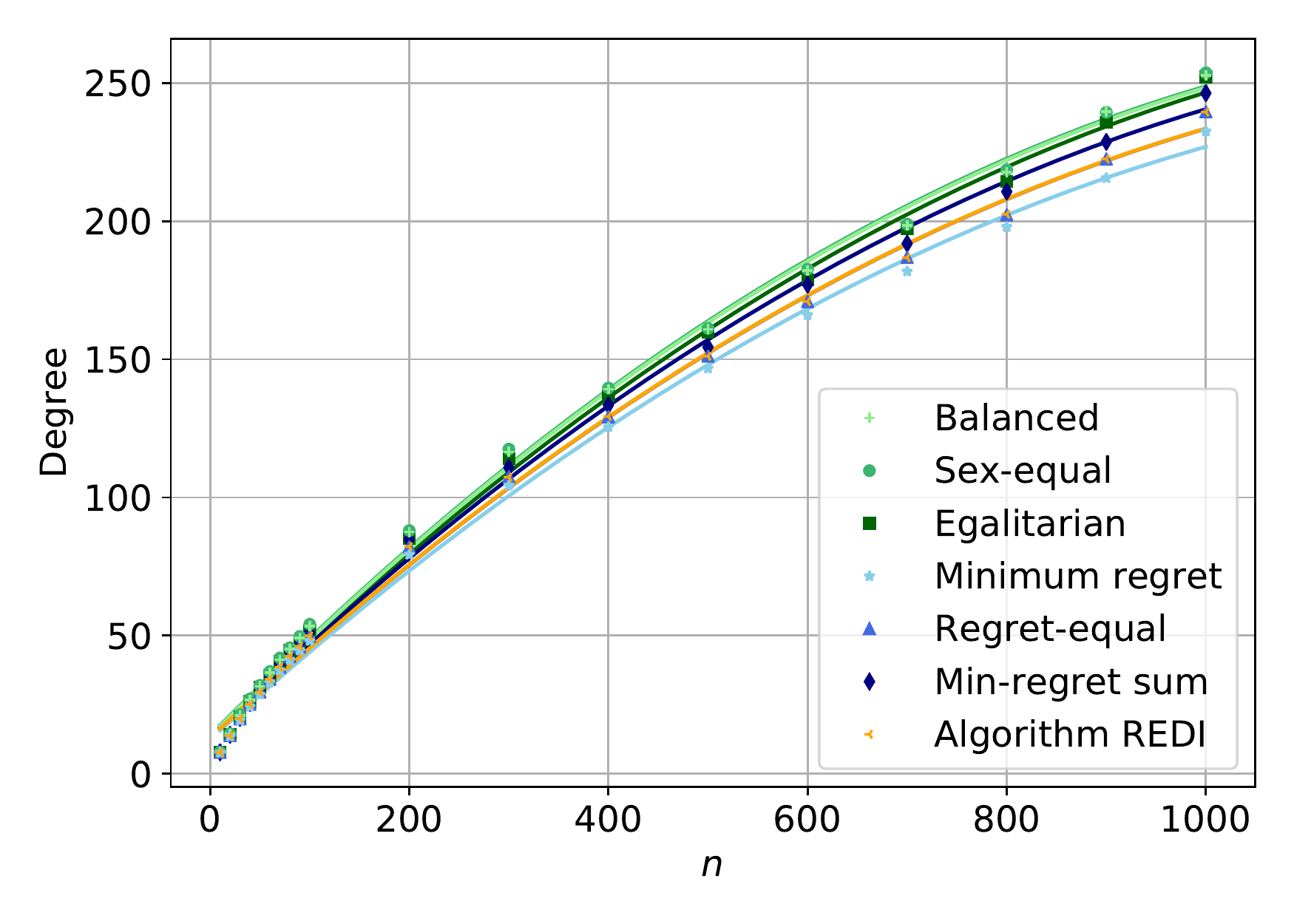}
    \caption{Plot of degree.}
    % \label{}
  \end{subfigure} 
  
  \begin{subfigure}[b]{0.48\textwidth}
  \centering
    \includegraphics[scale=0.36]{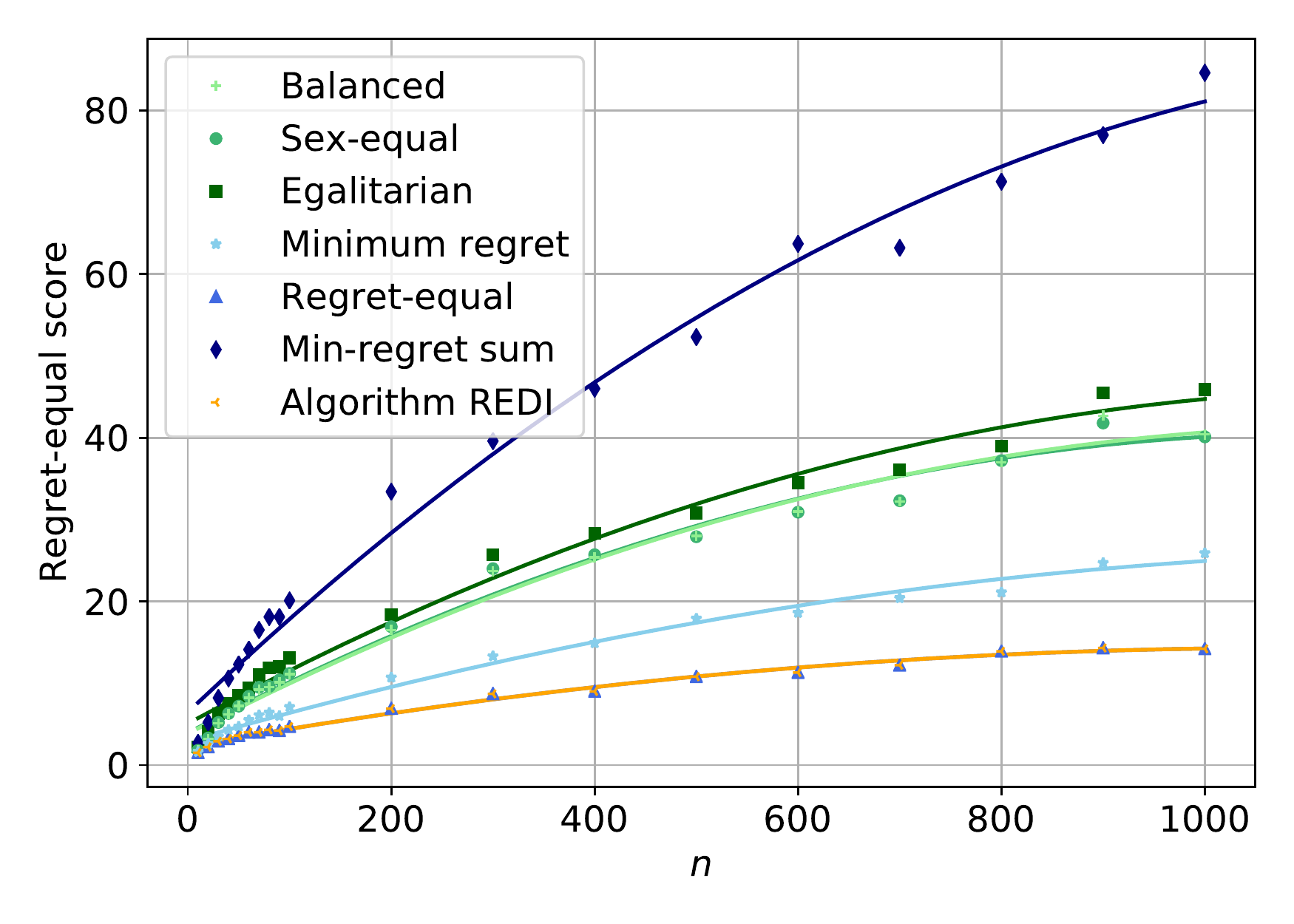}
    \caption{Plot of regret-equality score.}
    \label{reg-eq-opt-plots_re}
  \end{subfigure}
    \begin{subfigure}[b]{0.48\textwidth}
  \centering
    \includegraphics[scale=0.36]{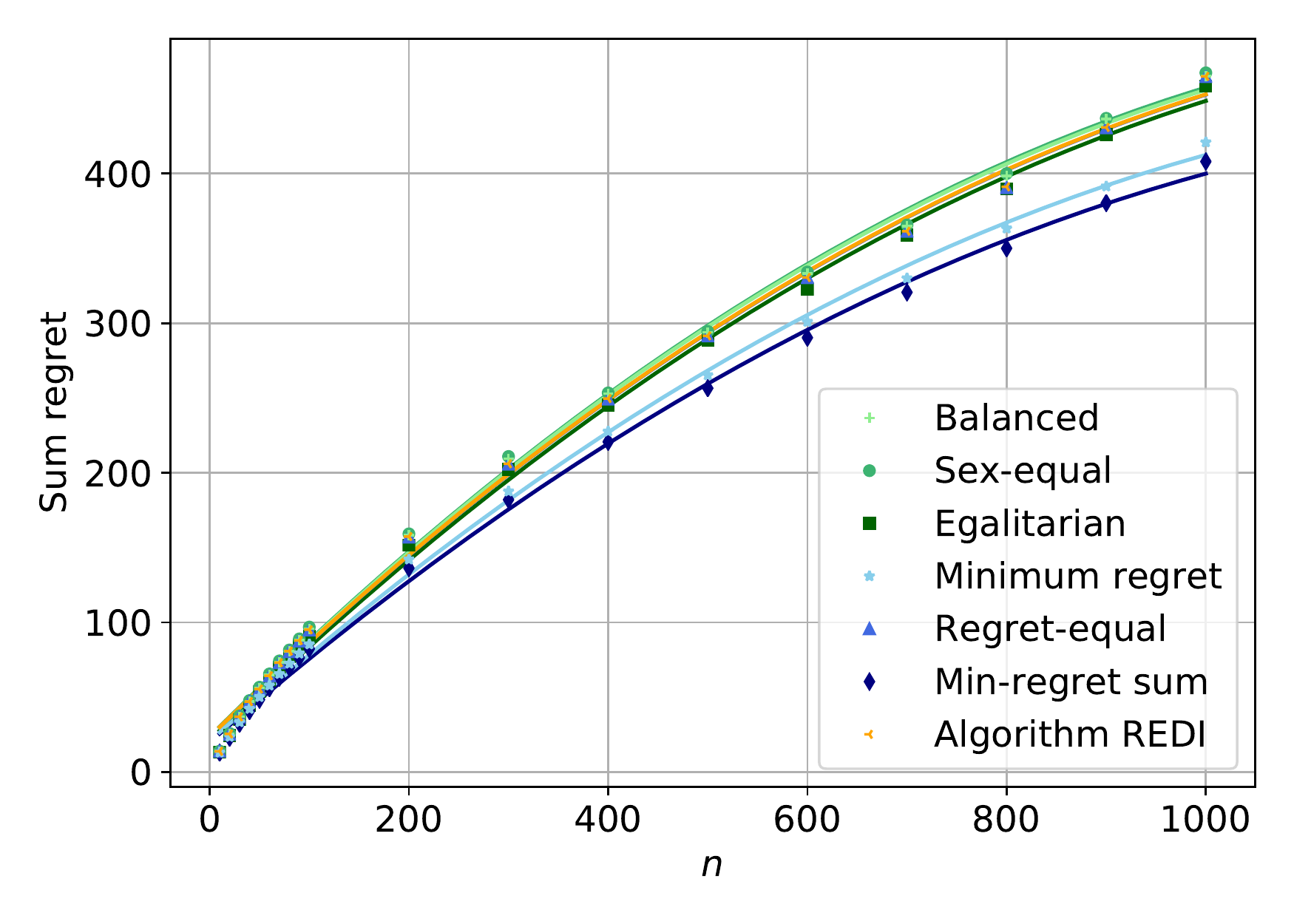}
    \caption{Plot of regret sum.}
    % \label{}
  \end{subfigure} 

  \caption{Plots of experiments to compare six different optimal stable matchings (balanced, sex-equal, egalitarian, min-regret, regret-equal, min-regret sum), and output from Algorithm REDI, over a range of measures (including balanced score, sex-equal score, cost, degree, regret-equality score, regret sum). A second order polynomial model has been assumed for all best-fit lines.}
  \label{reg-eq-opt-plots}
\end{figure}

\section{Future work}
\label{sm_re_fw_sec}

We introduced two new notions of fair stable matchings for \acrshort{smi}, namely, the regret-equal stable matching and the min-regret sum stable matching. We presented algorithms that are able to compute matchings of these types in polynomial time: $O(d_0 nm)$ time for the regret-equal stable matching, where $d_0 = |d_U(M_0) - d_W(M_0)|$; and $O(d_s m)$ time for the min-regret sum stable matching, where $d_s = d_U(M_z) - d_U(M_0)$. It remains open as to whether these time complexities can be improved. 

\newpage
\bibliography{bibliography}

\begin{thebibliography}{10}

\bibitem{AIM03}
D.J. Abraham, R.W. Irving, and D.F. Manlove.
\newblock The {S}tudent-{P}roject {A}llocation {P}roblem.
\newblock In {\em Proceedings of ISAAC '03: the 14th Annual International
  Symposium on Algorithms and Computation}, volume 2906 of {\em Lecture Notes
  in Computer Science}, pages 474--484. Springer, 2003.

\bibitem{AIM07}
D.J. Abraham, R.W. Irving, and D.F. Manlove.
\newblock Two algorithms for the {S}tudent-{P}roject allocation problem.
\newblock {\em Journal of Discrete Algorithms}, 5(1):79--91, 2007.
\newblock Preliminary version appeared as \cite{AIM03}.

\bibitem{BKMPABCDDHHJKKNSSVV19}
P.~Bir\'o, J.~van~de Klundert, D.~Manlove, W.~Pettersson, T.~Andersson,
  L.~Burnapp, P.~Chromy, P.~Delgado, P.~Dworczak, B.~Haase, A.~Hemke,
  R.~Johnson, X.~Klimentova, D.~Kuypers, A.~Nanni Costa, B.~Smeulders,
  F.~Spieksma, M.O. Valent\'in, and A.~Viana.
\newblock Modelling and optimisation in european kidney exchange programmes.
\newblock {\em European Journal of Operational Research}, 13(4):1--10, 2019.

\bibitem{Fed90}
T.~Feder.
\newblock {\em Stable Networks and Product Graphs}.
\newblock PhD thesis, Stanford University, 1990.
\newblock Published in \emph{Memoirs of the American Mathematical Society},
  vol.\ 116, no.\ 555, 1995.

\bibitem{GS62}
D.~Gale and L.S. Shapley.
\newblock College admissions and the stability of marriage.
\newblock {\em American Mathematical Monthly}, 69:9--15, 1962.

\bibitem{GS85}
D.~Gale and M.~Sotomayor.
\newblock Some remarks on the stable matching problem.
\newblock {\em Discrete Applied Mathematics}, 11:223--232, 1985.

\bibitem{GRSZ19}
S.~Gupta, S.~Roy, S.~Saurabh, and M.~Zehavi.
\newblock Balanced stable marriage: How close is close enough?
\newblock In {\em Proceedings of WADS '19: the 16th Algorithms and Data
  Structures Symposium}, Lecture Notes in Computer Science, pages 423--437.
  Springer, 2019.

\bibitem{Gus87}
D.~Gusfield.
\newblock Three fast algorithms for four problems in stable marriage.
\newblock {\em SIAM Journal on Computing}, 16(1):111--128, 1987.

\bibitem{GI89}
D.~Gusfield and R.W. Irving.
\newblock {\em The Stable Marriage Problem: Structure and Algorithms}.
\newblock MIT Press, 1989.

\bibitem{CPLEX}
{I}{B}{M}.
\newblock {C}{P}{L}{E}{X} optimizer.
\newblock \url{https://www.ibm.com/analytics/cplex-optimizer}, 2019.

\bibitem{IL86}
R.W. Irving and P.~Leather.
\newblock The complexity of counting stable marriages.
\newblock {\em SIAM Journal on Computing}, 15(3):655--667, 1986.

\bibitem{ILG87}
R.W. Irving, P.~Leather, and D.~Gusfield.
\newblock An efficient algorithm for the ``optimal'' stable marriage.
\newblock {\em Journal of the ACM}, 34(3):532--543, 1987.

\bibitem{Kat93}
A.~Kato.
\newblock Complexity of the sex-equal stable marriage problem.
\newblock {\em Japan Journal of Industrial and Applied Mathematics}, 10:1--19,
  1993.

\bibitem{Knu76}
D.E. Knuth.
\newblock {\em Mariages Stables}.
\newblock Les Presses de L'Universit\'{e} de Montr\'{e}al, 1976.
\newblock {E}nglish translation in \emph{Stable Marriage and its Relation to
  Other Combinatorial Problems}, volume 10 of CRM Proceedings and Lecture
  Notes, American Mathematical Society, 1997.

\bibitem{Man13}
D.F. Manlove.
\newblock {\em Algorithmics of Matching Under Preferences}.
\newblock World Scientific, 2013.

\bibitem{MI14}
E.~McDermid and R.W. Irving.
\newblock Sex-equal stable matchings: Complexity and exact algorithms.
\newblock {\em Algorithmica}, 68:545--570, 2014.

\bibitem{MOD11}
S.~Mitchell, M.~O'Sullivan, and I.~Dunning.
\newblock Pu{L}{P}: A linear programming toolkit for {P}ython.
\newblock {\em Optimization Online}, 2011.

\bibitem{PR95}
E.~Peranson and R.R. Randlett.
\newblock The {N}{R}{M}{P} matching algorithm revisited: Theory versus
  practice.
\newblock {\em Academic Medicine}, 70(6):477--484, 1995.

\bibitem{Tan11}
O.~Tange.
\newblock {G}{N}{U} parallel - the command-line power tool.
\newblock {\em The USENIX Magazine}, pages 42--47, 2011.

\end{thebibliography}

\newpage
\appendix
\begin{landscape}
\section{Algorithm REDI supplement}
This appendix section presents supplementary information for Algorithm REDI.

\subsection{Regret-equal degree tuples}
\label{sm_re_degpairs_supp}
Figure \ref{re_ex_differences} shows the possible degree pairs of a regret-equal stable matching in the general case. Figure \ref{re_ex_differences_2} shows the possible degree pairs for a regret-equal stable matching where $d(M_0)=(2, 6)$ and $n \geq 9$.

\begin{figure}[h]
\centering

% \begin{tabular}{p{3.5cm} | p{3cm} p{3.5cm} p{2cm} p{2cm} p{1cm} p{2cm} p{3cm} p{3cm}}
\begin{tabular}{l|l|l|l|l|l|l|l|l|l}
$r(M)$ & \multicolumn{7}{l}{Degree pairs $(d_U(M), d_W(M))$} \\
\hline
$k$ & $1$ & $2$ & $...$ & $b_0 - a_0$ & $b_0 - a_0 + 1$ & $b_0 - a_0 + 2$ & $...$ & $2b_0 - 2a_0 - 1$ & $2b_0 - 2a_0$ \\
\hline
$d_0=b_0 - a_0$ & $(a_0, b_0)$ & &&&&&&&\\
$b_0 - a_0 - 1$ & $(a_0, b_0 - 1)$ & $(a_0 + 1, b_0)$ &&&&&&\\
% $b_0 - a_0 - 2$ & $(a_0, b_0 - 2)$ & $(a_0 + 1, b_0 - 1)$ & $(a_0 + 2, b_0)$\\
$...$ & $...$ & $...$ &&&&&&& \\
$1$ & $(a_0, a_0 + 1)$ & $(a_0 + 1, a_0 + 2)$ & $...$ & $(b_0 - 1, b_0)$ &&&&&\\
$0$ & $(a_0, a_0)$ & $(a_0 + 1, a_0 + 1)$ & $...$ & $...$ & $(b_0, b_0)$ &&&&\\
$1$ & $(a_0, a_0 - 1)$ & $(a_0 + 1, a_0)$ & $...$ & $...$ & $...$ & $(b_0 + 1, b_0)$ &&\\
$...$ & $...$ & $...$ & $...$ & $...$ & $...$ & $...$ && \\
$b_0 - a_0 - 2$ & $(a_0, a_0 - d_0 + 2)$ & $(a_0 + 1, a_0 - d_0 + 3)$ & $...$ & $...$ & $...$ & $...$ & $...$ & $(2b_0 - a_0 - 2, b_0)$\\
$b_0 - a_0 - 1$ & $(a_0, a_0 - d_0 + 1)$ & $(a_0 + 1, a_0 - d_0 + 2)$ & $...$ & $...$ & $...$ & $...$ & $...$ & $...$ & $(2b_0 - a_0 - 1, b_0)$ \\

\end{tabular}
\caption{Possible regret-equal degree tuples for men and women when $d(M_0)=(a_0, b_0)$, $n \geq 2b_0 - a_0 - 1$ and $a_0 - d_0 + 1 \geq 1$.}
\label{re_ex_differences}
\end{figure}

\end{landscape}

\begin{figure}
\centering

\begin{tabular}{l|l|l|l|l|l|l|l|l}
$r(M)$ & \multicolumn{8}{l}{Degree pairs $(d_U(M), d_W(M))$} \\
\hline
$k$ & $1$ & $2$ & $3$ & $4$ & $5$ & $6$ & $7$  & $8$\\
\hline
$d_0=4$ & $(2, 6)$ &&&&&&\\
$3$ & $(2, 5)$ & $(3, 6)$ &&&&& \\
$2$ & $(2, 4)$ & $(3, 5)$ & $(4, 6)$ &&&&\\
$1$ & $(2, 3)$ & $(3, 4)$ & $(4, 5)$ & $(5, 6)$ &&&\\
$0$ & $(2, 2)$ & $(3, 3)$ & $(4, 4)$ & $(5, 5)$ & $(6, 6)$&&&\\
$1$ & $(2, 1)$ & $(3, 2)$ & $(4, 3)$ & $(5, 4)$ & $(6, 5)$ & $(7, 6)$ &&\\
$2$ &  & $(3, 1)$ & $(4, 2)$ & $(5, 3)$ & $(6, 4)$ & $(7, 5)$ & $(8, 6)$\\
$3$ &  &  & $(4, 1)$ & $(5, 2)$ & $(6, 3)$ & $(7, 4)$ & $(8, 5)$ & $(9, 6)$\\
\end{tabular}
\caption{Possible regret-equal degree tuples for men and women when $d(M_0)=(2, 6)$ and $n \geq 9$.}
\label{re_ex_differences_2}
\end{figure}

\subsection{Correctness proof}
\label{app_redi_corr_proof}
Theorem \ref{sm_re_app_alg_redi_thm}, with accompanying Proposition \ref{sm_re_alg_degree_prop}, shows that Algorithm REDI always produces a regret-equal stable matching. Theorem \ref{sm_re_app_alg_redi_time} shows that Algorithm REDI runs in $O(d_0nm)$ time. 

\begin{proposition}
	Let $I$ be an instance of \acrshort{smi} and let $M$ and $M'$ be stable matchings in $I$ where for each man $m_i \in U$, $\text{rank}(m_i, M(m_i)) \leq \text{rank}(m_i, M'(m_i))$ and $d_U(M) = d_U(M')$. Let $Q$ and $Q'$ denote the set of rotations eliminated from $M_0$ to reach $M$ and $M'$ respectively. Then stable matching $M''$ with $d(M'') = d(M')$ may be found by ensuring all rotations in $R_d = c(\{\rho \in R : \exists (m,w) \in \rho$ where $d_W(M) \geq \text{rank}(w, m) > d_W(M')\}\backslash Q)$ are eliminated from $M$. Figure \ref{sm_re_correctness_fig} shows a summary of degree pairs for $M$, $M'$ and $M''$ in the general case.
	\label{sm_re_alg_degree_prop}
	
	\begin{figure}
	\centering
	\begin{tabular}{l}	
$(a_0 + k - 1, b_0)$\\
$(a_0 + k - 1, b_0 - 1)$\\
$(a_0 + k - 1, b_0 - 2)$\\
$...$\\
$d(M)$\\
$...$\\
$d(M') = d(M'')$\\
$...$\\
$(a_0 + k - 1, \max \{a_0 - d_0 + k, 1\})$\\
	\end{tabular}
	\caption{Degree pairs in column $k = d_U(M)$ for instance $I$, built as per the description in Section \ref{sm_re_doa}.}
	\label{sm_re_correctness_fig}
	\end{figure}
	
\end{proposition}
\begin{proof}
	First, since $\text{rank}(m_i, M(m_i)) \leq \text{rank}(m_i, M'(m_i))$ and each rotation must make some man worse off, we know that each rotation in $Q$ must be eliminated to reach $M'$ and therefore $Q \subseteq Q'$. Second, we also know that any rotation containing a woman at rank larger than $d_W(M')$ must be eliminated from $M$ in order to reach $M'$ (additional rotations may also have been eliminated). But these are precisely the rotations in $R_d$, hence $R_d \subseteq Q'$. Let $M''$ be the stable matching found when eliminating $R_d$ on $M$ and let $Q'' = Q \cup R_d$ denote the unique set of rotations corresponding to $M''$. Then $Q'' \subseteq Q'$ since $Q \subseteq Q'$ and $R_d \subseteq Q'$.
	
	We observe the following.
	\begin{itemize}
		\item Since $Q \subseteq Q''$, it must be that $d_U(M) \leq d_U(M'')$, and as $d_U(M) = d_U(M')$, it follows that $d_U(M') \leq d_U(M'')$;
		\item $Q \cup R_d$ is the set of all rotations with pairs containing women of rank larger than $d_W(M')$. Since $Q'' = Q \cup R_d$ and all rotations in $Q''$ are eliminated on $M_0$ to reach $M''$, it must be the case that $d_W(M'') \leq d_W(M')$; 
%		\item Since $Q \subseteq Q''$ we know that $d_W(M'') \leq d_W(M)$ and since $R_d \subseteq Q''$, it must additionally hold that $d_W(M'') \leq d_W(M')$;
		\item Since $Q'' \subseteq Q'$, it must be that $d_U(M'') \leq d_U(M')$ and $d_W(M'') \geq d_W(M')$.
	\end{itemize}
	
	 Hence $d(M'') = d(M')$ as required.
\end{proof}

\smreappalgredithm
\begin{proof}
There are four points in Algorithm \ref{alg_sm_reg_eq}'s execution where we return a matching, namely Lines \ref{alg_sm_reg_eq_M0_return}, \ref{alg_sm_reg_eq_first_return}, \ref{alg_sm_reg_eq_second_return} and \ref{alg_sm_reg_eq_third_return}. First we show that if $M$ is a matching returned at any of these points then $M$ is stable. Next we look at each of these points where a matching may be returned and show they are regret-equal stable matchings.

Let $M$ be the matching returned at any of the four points above. Then $M$ is stable since it is found by iteratively eliminating sets of rotations that form closed subsets of the rotation poset of $I$ (on Line \ref{alg_sm_reg_eq_elimination_1} of Algorithm \ref{alg_sm_reg_eq} and Line \ref{alg_sm_reg_eq_sub_elimination_2} of Algorithm \ref{alg_sm_reg_eq_sub}) starting from the man-optimal stable matching (created on Line \ref{alg_sm_reg_eq_man-opt}). Since there is a $1$-$1$ correspondence between closed subsets of the rotation poset and the set of all stable matchings \cite[Theorem 3.1]{ILG87}, $M$ is a stable matching.

Let $M_1$ be a matching that is returned on Line \ref{alg_sm_reg_eq_M0_return}. Since $M_1$ has been returned on Line \ref{alg_sm_reg_eq_M0_return} it must be that $d_U(M_1) \geq d_W(M_1)$ and therefore by the same reasoning given in the second paragraph of Section \ref{sm_re_doa}, $M_1$ is a regret-equal stable matching.

Let $M_2$ be a matching that is returned by either Line \ref{alg_sm_reg_eq_first_return} or Line \ref{alg_sm_reg_eq_second_return}. To be returned at these points $r(M_2) = 0$ and therefore $M_2$ is a regret-equal stable matching.

Let $M_3$ be a matching that is returned on Line \ref{alg_sm_reg_eq_third_return} and let $M'$ be a regret-equal stable matching such that $d_U(M')$ is minimum over all regret-equal stable matchings. We will prove that $r(M_3)=r(M')$ by showing it will not have been possible for us to miss a stable matching with regret-equality score equal to $r(M')$ during the algorithm's execution. First we show that the column operation (Algorithm \ref{alg_sm_reg_eq_sub}) must be executed for column $d_U(M')$. Then, we show that during this column operation $M_{opt}$ will be updated such that $r(M_{opt}) = r(M')$.

 Since $M'$ is not returned on Line \ref{alg_sm_reg_eq_M0_return}, $M' \neq M_0$. If $d_U(M') = d_U(M_0)$ then clearly the column operation is executed on Line \ref{alg_sm_reg_eq_subcall_1} for column $d_U(M')$. Assume therefore that $d_U(M') \neq d_U(M_0)$. Without loss of generality let $\{m_1,m_2...,m_k\}$ be the set of men who are at rank $d_U(M')$ in $M'$. We enter the while loop on Line \ref{alg_sm_reg_eq_while} for each man $m_j \in \{m_1,m_2...,m_k\}$. $M$ is initialised to $M_0$. Successively, the algorithm eliminates all rotations in $c(\phi(m_j, M(m_j)))$ that are not yet eliminated until $M(m_j) = M'(m_j)$. This must be possible since $(m_j, M'(m_j)) \in M'$ and $M'$ is a stable matching in $I$. During the while loop iteration that rotates $m_j$ down to their $M'(m_j)$ partner, it is not necessarily the case that the current matching $M$ updates its man-degree to $d_U(M')$ at this point. This is because although $d_U(M) = d_U(M')$, an earlier movement of $m_j$ in a previous while loop iteration may have brought some other man in $\{m_1,m_2...,m_k\}$ down to their partner in $M'$ already. Note that it is not possible for a man to overshoot column $d_U(M')$ when moving $m_j$ down to their $M'(m_j)$ partner since the set of rotations we have eliminated to bring $m_i$ down to $M'(m_i)$ is a subset of the rotations corresponding to $M'$. However, for at least one of these men $m_i\in \{m_1,m_2...,m_k\}$, the while loop iteration that moves $m_i$ down to $M'(m_i)$ will also lead to $d_U(M)$ increasing to the same value as $d_U(M')$. Assume we are beginning the while loop iteration on Line \ref{alg_sm_reg_eq_while} for man $m_i$. Let $M = M_0$. We continue eliminating rotations that have not yet been eliminated in $c(\phi(m_i, M(m_i)))$ until $M(m_i) = M'(m_i)$. Our choice of $m_i$ ensures that the movement of $m_i$ to their $M'(m_i)$ partner occurs at the same time as $d_U(M)$ increases to $d_U(M')$ and so we satisfy the conditions on Line \ref{sm_re_alg_condition_for_colop} to perform the column operation (Algorithm \ref{alg_sm_reg_eq_sub}) on $M$ for column $d_U(M')$. 
 
From above we know that the column operation (Algorithm \ref{alg_sm_reg_eq_sub}) will be executed for column $d_U(M')$. Either we start this column operation with $M_0$, or we have only eliminated the minimum number of rotations necessary to take $m_i$ down to $M'(m_i)$ from $M_0$. In either case we know that $\text{rank}(m_l, M(m_l)) \leq \text{rank}(m_l, M'(m_l))$ for all $m_l \in U$.
% and so it must be that $d_W(M) \geq d_W(M')$. 
%Also, since $(m_i, M(m_i))$ exists in $M'$, there is at least one regret-equal stable matching $(M')$ that can be reached by eliminating rotations from $M$ (without increasing the men's degree).
%If $d_W(M) = d_W(M')$ then $M_{opt}$ will get set to $M$ on Line \ref{alg_sm_reg_eq_sub_setToMopt} of Algorithm \ref{alg_sm_reg_eq_sub}. But then since $r(M) < r(M_3)$, $M_{opt}$ would never be updated to $M_3$ to be returned, a contradiction.
%Suppose then that $d_W(M) > d_W(M')$. 
For this column a regret-equal stable matching may be found with degree pairs that are either $(d_U(M'), d_U(M') + r(M'))$ or $(d_U(M'), d_U(M') - r(M'))$. The degree pair of matching $M'$ is given by one of the above pairs, but it may be possible for regret-equal stable matchings to exist in $I$ with both of the above degree pairs. 
Assume firstly that $d(M') = (d_U(M'), d_U(M') + r(M'))$. 
%Assume that regret-equal stable matching $M''$ exists with $d(M'') = (d_U(M'), d_U(M') + r(M'))$. 
The algorithm will attempt to successively eliminate from $M$ rotations containing women of rank $d_W(M)$. By Proposition \ref{sm_re_alg_degree_prop}, since $\text{rank}(m_i, M(m_i)) \leq \text{rank}(m_i, M'(m_i))$ for all $m_i \in U$ and $d_U(M) = d_U(M')$, we know the algorithm will continue this process until $M$ is updated to a regret-equal stable matching with $d(M)=d(M')$ and so $M_{opt}$ will be set to $M$ with $r(M_{opt})=r(M')$.
%We know that this process cannot have increased $d_U(M)$ since $M'$ exists with $d_U(M')=d_U(M)$ and $d_W(M') \leq d_W(M)$. 
Assume then that $d(M') = (d_U(M'), d_U(M') - r(M'))$, where a regret-equal stable matching may or may not exist with degree pair $(d_U(M'), d_U(M') + r(M'))$.
% and stable matching $M'''$ exists with $d(M''') = (d_U(M'), d_U(M') - r(M'))$. 
Then similar to before, the algorithm successively eliminates from $M$ rotations containing women of rank $d_W(M)$. This may result in a regret-equal stable matching $M$ being found with $d(M) = (d_U(M'), d_U(M') + r(M'))$ in which case $M_{opt}$ will be set to $M$ with $r(M_{opt})=r(M')$. Assume this is not the case. Then, by Proposition \ref{sm_re_alg_degree_prop}, since $\text{rank}(m_i, M(m_i)) \leq \text{rank}(m_i, M'(m_i))$ for all $m_i \in U$ and $d_U(M) = d_U(M')$, we know the algorithm will continue eliminating rotations containing women of rank $d_W(M)$ until $M$ is updated to a regret-equal stable matching with $d(M)=d(M')$ and so $M_{opt}$ will be set to $M$ with $r(M_{opt})=r(M')$ as above.
%Similar to above, this process cannot have increased $d_U(M)$ since $M'$ must exist with $d_U(M')=d_U(M)$ and $d_W(M') = d_W(M)$, so $M$ is a regret-equal stable matching, contradicting the fact that $M_3$ is returned. 

Therefore any matching returned by Algorithm REDI is a regret-equal equal stable matching.
\end{proof}

\smreappalgreditime
\begin{proof}
The for loop on Line \ref{alg_sm_reg_eq_for} of Algorithm \ref{alg_sm_reg_eq} iterates over all men, $n$ times, where $n$ is the number of men or women. During the nested while loop on Line \ref{alg_sm_reg_eq_while}, each man $m_i$ may be rotated down his preference list on Line \ref{alg_sm_reg_eq_elimination_1}, $2d_0 - 1$ times (the maximum possible number of columns from Figure \ref{re_ex_differences} minus $1$). Rotations are eliminated on Line \ref{alg_sm_reg_eq_elimination_1} of Algorithm \ref{alg_sm_reg_eq} and Line \ref{alg_sm_reg_eq_sub_elimination_2} of Algorithm \ref{alg_sm_reg_eq_sub} successively, beginning at the man-optimal stable matching $M_0$, meaning $O(m)$ rotations are eliminated in total for each while loop iteration at Line \ref{alg_sm_reg_eq_while} of Algorithm \ref{alg_sm_reg_eq}. This may also be viewed as $O(m)$ man-woman pair changes since the number of possible man-woman pairs in $I$ is $O(m)$ and each pair existing in the set of rotations is unique. Therefore Algorithm REDI runs in $O(d_0nm)$ time and since preference lists of men and women are finite, the algorithm terminates.
\end{proof}

\section{Algorithm MRS correctness and time complexity proof}
\label{app_mrs_corr_tc_proof}
Theorem \ref{sm_msr_alg_proof} provides the correctness proof and time complexity analysis for Algorithm MRS. In Lemma \ref{sm_re_sp_setstable} and Theorem \ref{sm_msr_alg_proof} we will use the following notation and terminology. Let $I$ be an instance of \acrshort{smi}. Let $I_T$ be the truncated instance of $I$, created on Line \ref{sm_msr_alg_trunc} of Algorithm \ref{sm_mrs_alg}, where men are truncated below rank $a$. Let $\mathcal{M}_T$ be the set of stable matchings in $I_T$. Finally, let 
$\text{reduced}(\mathcal{M}) = \{M \in \mathcal{M} : \forall (m_i, w_j) \in M, \text{rank}(m_i, w_j) \leq a\}$.

\begin{lemma}
\label{sm_re_sp_setstable}
	$\mathcal{M}_T = \text{reduced}(\mathcal{M})$. 
\end{lemma}
\begin{proof}
Let stable matching $M'$ of size $n$ exist in $\mathcal{M}_T$. If we transform $I_T$ to $I$ then we are adding preference list pairs $(m_i, w_j)$ where $\text{rank}(m_i, w_j) > \text{rank}(m_i, M'(m_i))$. But then, $(m_i, w_j)$ cannot constitute a blocking pair and so $M'$ must be stable in $I$ with $M' \in \mathcal{M}$. Also, since $M'$ is in $I_T$ it must be the case that $\text{rank}(m_i, w_j) \leq a$. Hence $\mathcal{M}_T \subseteq \text{reduced}(\mathcal{M})$.

Let stable matching $M''$ exist in $\text{reduced}(\mathcal{M})$. By the definition of $\text{reduced}(\mathcal{M})$ and $I_T$ all pairs in $M''$ must exist in preference lists of the truncated instance $I_T$. If we transform $I$ to $I_T$ we will only be removing some pairs from preference lists that do not exist in $M''$. Therefore since we are only removing pairs, it is not possible to introduce pairs into $I_T$ that would block $M''$ and so $M''$ is stable in $I_T$. Hence $\mathcal{M}_T \supseteq \text{reduced}(\mathcal{M})$.

Therefore $\mathcal{M}_T = \text{reduced}(\mathcal{M})$, as required.
\end{proof}

\smmsralgproof
\begin{proof}
Let $M$ be a min-regret sum stable matching in $I$ with $d_U(M)$ minimum among all min-regret sum stable matchings. We show that it is not possible to miss a stable matching $M'$ with $d(M') = d(M)$, during the execution of Algorithm MRS. On Line \ref{sm_mrs_alg_while} of Algorithm \ref{sm_mrs_alg}, we iterate over all possible men's degrees that may correspond to a min-regret sum stable matching, and will enter the while loop for degree value $a = d_U(M)$. Since $a = d_U(M)$, we know that $M \in \text{reduced}(\mathcal{M})$ and so by Lemma \ref{sm_re_sp_setstable}, $M \in \mathcal{M}_T$. We also know by Lemma \ref{sm_re_sp_setstable}, that $M_z^T \in \mathcal{M}$. Let $R'$ and $R_z^T$ be the set of rotations associated with $M$ and $M_z^T$ in $I$. On Line \ref{sm_msr_alg_trunc_wopt} of the algorithm we find the woman-optimal stable matching $M_z^T$ in $\mathcal{M}_T$ and so we must have $\text{rank}(w_j, M_z^T(w_j)) \leq \text{rank}(w_j, M(w_j))$, for all women $w_j$, since $M \in \mathcal{M}_T$. From this inequality, we know $R' \subseteq R_z^T$, and since $d_U(M)=a$ it must be that $d_U(M_z^T) = a$. Additionally, this inequality implies $d_W(M_z^T) \leq d_W(M)$. As $M$ is a min-regret sum stable matching, there cannot be a stable matching with man degree $a$ and woman degree $< d_W(M)$, but as $d_W(M_z^T) \leq d_W(M)$, we must have $d_W(M_z^T) = d_W(M)$. Hence, $d(M_z^T) = d(M)$. Note that due to the choice of $M$, with $d_U(M)$ minimum among all min-regret sum stable matchings, it is not possible for $M_{opt}$ to be updated to a min-regret stable matching prior to this point in the algorithm. Therefore, $M_{opt}$ is updated to $M_z^T$ on Line \ref{sm_mrs_alg_mopt_updated}. Additionally, as $M_{opt}$ is now a min-regret stable matching it will not be possible for it to be updated to another stable matching after this point, and so $M_{opt}$ is returned on Line \ref{sm_msr_alg_return_mopt}, as required.

% we know that the woman-optimal stable matching in the truncated instance $M_z^T$ must exist in $\text{reduced}(\mathcal{M})$ and that $M_z^T$ is woman-optimal in $\text{reduced}(\mathcal{M})$. Therefore, , .

Truncation of our instance requires only a single pass through preference lists of men and women and is therefore an $O(m)$ operation. The man-optimal and woman-optimal stable matchings can found in $O(m)$ time both within and outwith the while loop. Since the number of potential men's degrees that we iterate over is bounded by $d_s = d_U(M_z) - d_U(M_0)$, we have an overall time complexity of $O(d_s m)$ for Algorithm MRS.
\end{proof}

\section{Experiments supplement}
\label{sm_re_exps_sec_app}

\subsection{Correctness testing}
\label{sm_re_exps_sec_correctness_app}

Correctness tests were run in the following way. In addition to the $19$ generated instance types described in Section \ref{sm_re_exps_sec}, a further two instance types were generated where $n \in \{6, 8\}$, with $5000$ instances for each type. Over all instances of the $21$ instance types, each matching output by an algorithm (one for Algorithm REDI, and multiple for the Algorithm ENUM), was tested for \emph{(1) capacity:} each man (woman) may only be assigned to one woman (man) respectively; and \emph{(2) stability:} no blocking pair exists. Additionally, the regret-equality score of the stable matchings output by each of the algorithms were compared against each other to ensure they were identical values. These tests were were written in Java and compiled using Java version $1.8.0$. Finally, for all instances types where $n \leq 50$, further correctness testing was conducted on the Algorithm ENUM to ensure that the correct number of stable matchings was produced. This was done using an \acrlong{ip} (\acrshort{ip}) model built using the \acrshort{ip} modelling framework PuLP (version $1.6.9$) \cite{MOD11} running CPLEX (version $12.8.0$) \cite{CPLEX} with Python version $2.7.15$. Similar to above, all instances were run on a single thread with $16$ instances run in parallel using GNU Parallel \cite{Tan11}. A timeout time of $30$ minutes was applied to each instance for the PuLP program, and all instances completed within the time limit. All correctness tests passed successfully.

\subsection{Experiments figures and tables}
\label{sm_re_exps_sec_f_t_app}
\FloatBarrier
This appendix section presents tables referred to in Section \ref{sm_re_exps_sec}. Instance types are labelled according to $n$, e.g., $S100$ is the instance type containing instances where $n=100$. Table \ref{sm_re_table_duration} provides data for Figure \ref{reg-eq-duration}. Tables \ref{sm_re_table_balancedScoreAv}, \ref{sm_re_table_sexEqualCostAv},  \ref{sm_re_table_egalCostAv}, \ref{sm_re_table_degreeAv}, \ref{sm_re_table_regretEqualScoreAv} and \ref{sm_re_table_sumRegretAv} provide data for the plots in Figure \ref{reg-eq-opt-plots}.

\begin{table}[h!] \centerline{\begin{tabular}{ R{1.2cm} | R{1.4cm} R{1.4cm} R{1.4cm} R{1.4cm} | R{1.4cm} R{1.6cm} R{1.4cm} R{1.4cm} }\hline\hline Case & REDI$_{av}$ & REDI$_{med}$ & REDI$_{5}$ & REDI$_{95}$ & ENUM$_{av}$ & ENUM$_{med}$ & ENUM$_{5}$ & ENUM$_{95}$ \\ 
\hline S10 & $204.9$ & $206.0$ & $180.0$ & $226.0$ & $244.8$ & $244.0$ & $215.9$ & $270.0$ \\ 
 S20 & $214.3$ & $215.0$ & $191.0$ & $236.0$ & $256.1$ & $257.0$ & $227.9$ & $283.0$ \\ 
 S30 & $220.7$ & $221.5$ & $195.0$ & $243.0$ & $264.0$ & $265.5$ & $236.0$ & $288.0$ \\ 
 S40 & $234.1$ & $235.0$ & $208.0$ & $255.0$ & $279.9$ & $281.0$ & $248.0$ & $305.0$ \\ 
 S50 & $237.1$ & $238.0$ & $213.0$ & $259.0$ & $288.7$ & $289.0$ & $260.0$ & $317.0$ \\ 
 S60 & $252.5$ & $253.5$ & $229.0$ & $276.0$ & $311.3$ & $312.0$ & $281.0$ & $340.0$ \\ 
 S70 & $261.3$ & $262.0$ & $234.0$ & $284.0$ & $320.5$ & $321.0$ & $285.0$ & $353.0$ \\ 
 S80 & $269.7$ & $270.0$ & $246.0$ & $294.0$ & $334.5$ & $334.0$ & $296.9$ & $373.0$ \\ 
 S90 & $288.4$ & $290.0$ & $262.0$ & $312.0$ & $360.9$ & $360.0$ & $325.0$ & $400.0$ \\ 
 S100 & $290.7$ & $291.0$ & $259.9$ & $321.0$ & $369.3$ & $369.5$ & $329.0$ & $412.0$ \\ 
 S200 & $362.4$ & $354.0$ & $315.0$ & $441.0$ & $533.5$ & $514.0$ & $420.9$ & $695.3$ \\ 
 S300 & $447.1$ & $440.0$ & $425.0$ & $494.0$ & $789.1$ & $751.0$ & $606.9$ & $1047.1$ \\ 
 S400 & $541.3$ & $540.0$ & $520.0$ & $562.0$ & $1278.7$ & $1127.0$ & $891.9$ & $2136.6$ \\ 
 S500 & $697.6$ & $698.0$ & $662.9$ & $734.0$ & $1961.7$ & $1761.5$ & $1201.4$ & $3156.3$ \\ 
 S600 & $965.6$ & $968.0$ & $925.9$ & $1010.0$ & $2835.6$ & $2546.0$ & $1767.9$ & $4707.9$ \\ 
 S700 & $1154.3$ & $1160.0$ & $1080.0$ & $1219.0$ & $3986.8$ & $3194.5$ & $2218.7$ & $7620.5$ \\ 
 S800 & $1396.7$ & $1402.0$ & $1312.0$ & $1485.0$ & $5708.3$ & $4496.5$ & $2801.7$ & $12450.0$ \\ 
 S900 & $1701.9$ & $1712.0$ & $1585.9$ & $1836.0$ & $9043.3$ & $5958.0$ & $3439.9$ & $20366.2$ \\ 
 S1000 & $2087.9$ & $2112.0$ & $1855.0$ & $2293.0$ & $11802.0$ & $8092.5$ & $4234.8$ & $28824.8$ \\ 
 \hline\hline \end{tabular}} \caption{A comparison of time taken to execute Algorithm REDI              and Algorithm ENUM. Here REDI$_{av}$, REDI$_{med}$, REDI$_{5}$ and REDI$_{95}$ represent the mean,              median, $5$th percentile and $95$th percentile of Algorithm REDI for a given instance type.              Similar notation is used for Algorithm ENUM. Times are in ms.} \label{sm_re_table_duration} \end{table} 

\begin{table}[] \centerline{\begin{tabular}{ R{1.2cm} | R{1.8cm} R{1.8cm} R{1.8cm} R{1.8cm} R{1.8cm} R{2cm} R{1.8cm} }\hline\hline Case & Balanced & Sex-equal & Egalitarian & Minimum regret & Regret-equal & Min-regret sum & Algorithm REDI  \\ 
\hline S10 & $32.1$ & $32.1$ & $33.3$ & $32.8$ & $32.8$ & $35.6$ & $33.4$ \\ 
 S20 & $90.5$ & $90.9$ & $95.1$ & $94.2$ & $94.0$ & $104.5$ & $96.9$ \\ 
 S30 & $165.6$ & $166.2$ & $177.5$ & $174.9$ & $176.5$ & $199.4$ & $181.5$ \\ 
 S40 & $254.9$ & $255.8$ & $273.0$ & $270.7$ & $270.2$ & $312.7$ & $278.4$ \\ 
 S50 & $357.2$ & $358.3$ & $382.0$ & $378.9$ & $379.5$ & $439.8$ & $393.0$ \\ 
 S60 & $466.4$ & $467.7$ & $495.5$ & $495.5$ & $496.5$ & $573.3$ & $516.2$ \\ 
 S70 & $588.8$ & $590.8$ & $626.5$ & $626.2$ & $629.1$ & $739.3$ & $651.4$ \\ 
 S80 & $720.1$ & $722.3$ & $769.7$ & $764.9$ & $769.9$ & $901.8$ & $798.9$ \\ 
 S90 & $861.1$ & $863.3$ & $921.7$ & $906.6$ & $914.9$ & $1054.9$ & $952.7$ \\ 
 S100 & $1004.7$ & $1007.2$ & $1073.4$ & $1062.7$ & $1063.1$ & $1245.0$ & $1103.2$ \\ 
 S200 & $2844.8$ & $2849.2$ & $3000.2$ & $3014.6$ & $3008.9$ & $3553.4$ & $3134.9$ \\ 
 S300 & $5224.1$ & $5230.1$ & $5510.4$ & $5488.6$ & $5550.5$ & $6348.2$ & $5741.3$ \\ 
 S400 & $8036.5$ & $8045.4$ & $8471.1$ & $8503.0$ & $8541.4$ & $9743.7$ & $8835.2$ \\ 
 S500 & $11215.7$ & $11223.6$ & $11740.4$ & $11891.9$ & $11857.3$ & $13577.3$ & $12352.0$ \\ 
 S600 & $14757.4$ & $14770.0$ & $15423.7$ & $15474.0$ & $15543.2$ & $18188.9$ & $16061.1$ \\ 
 S700 & $18576.7$ & $18590.2$ & $19407.7$ & $19525.4$ & $19553.7$ & $22512.5$ & $20217.5$ \\ 
 S800 & $22718.1$ & $22731.4$ & $23707.2$ & $23851.3$ & $23975.3$ & $27652.9$ & $24824.3$ \\ 
 S900 & $27098.2$ & $27113.3$ & $28198.3$ & $28678.0$ & $28667.4$ & $32719.2$ & $29707.8$ \\ 
 S1000 & $31684.8$ & $31702.0$ & $32976.9$ & $33364.7$ & $33393.2$ & $38599.1$ & $34551.0$ \\ 
 \hline\hline \end{tabular}} \caption{Mean balanced score for six different optimal stable matchings and output from Algorithm REDI.} \label{sm_re_table_balancedScoreAv} \end{table} 
\clearpage
\begin{table}[] \centerline{\begin{tabular}{ R{1.2cm} | R{1.8cm} R{1.8cm} R{1.8cm} R{1.8cm} R{1.8cm} R{2cm} R{1.8cm} }\hline\hline Case & Balanced & Sex-equal & Egalitarian & Minimum regret & Regret-equal & Min-regret sum & Algorithm REDI  \\ 
\hline S10 & $5.4$ & $5.3$ & $8.9$ & $7.0$ & $6.5$ & $12.2$ & $7.6$ \\ 
 S20 & $10.5$ & $10.0$ & $22.7$ & $17.8$ & $16.7$ & $36.2$ & $21.6$ \\ 
 S30 & $13.9$ & $13.1$ & $43.7$ & $32.1$ & $32.6$ & $75.8$ & $41.2$ \\ 
 S40 & $17.4$ & $16.1$ & $63.2$ & $48.8$ & $44.0$ & $121.8$ & $58.7$ \\ 
 S50 & $22.6$ & $21.5$ & $84.3$ & $65.1$ & $63.2$ & $169.8$ & $87.3$ \\ 
 S60 & $25.2$ & $23.1$ & $98.8$ & $80.7$ & $77.8$ & $213.6$ & $113.1$ \\ 
 S70 & $27.9$ & $25.5$ & $121.4$ & $99.8$ & $98.8$ & $290.9$ & $139.2$ \\ 
 S80 & $29.8$ & $27.6$ & $149.6$ & $118.0$ & $118.5$ & $347.5$ & $170.1$ \\ 
 S90 & $36.8$ & $33.5$ & $182.8$ & $125.6$ & $134.0$ & $378.8$ & $201.5$ \\ 
 S100 & $36.0$ & $32.6$ & $200.8$ & $150.1$ & $141.1$ & $456.9$ & $213.7$ \\ 
 S200 & $71.9$ & $66.4$ & $440.8$ & $399.6$ & $369.1$ & $1306.0$ & $598.0$ \\ 
 S300 & $101.9$ & $92.1$ & $762.3$ & $616.6$ & $709.7$ & $2067.9$ & $1060.2$ \\ 
 S400 & $135.3$ & $126.2$ & $1117.8$ & $1017.1$ & $1061.0$ & $3131.9$ & $1595.2$ \\ 
 S500 & $154.8$ & $142.5$ & $1345.4$ & $1431.4$ & $1343.2$ & $4322.4$ & $2255.0$ \\ 
 S600 & $188.3$ & $173.4$ & $1694.3$ & $1562.5$ & $1641.7$ & $6111.8$ & $2596.3$ \\ 
 S700 & $223.3$ & $207.4$ & $2071.3$ & $2014.9$ & $2053.9$ & $7077.1$ & $3268.0$ \\ 
 S800 & $233.0$ & $215.2$ & $2437.2$ & $2388.6$ & $2597.4$ & $8836.2$ & $4164.2$ \\ 
 S900 & $256.0$ & $240.1$ & $2705.8$ & $3227.4$ & $3179.2$ & $10125.8$ & $5077.3$ \\ 
 S1000 & $284.0$ & $265.0$ & $3147.4$ & $3438.4$ & $3454.8$ & $12400.0$ & $5583.6$ \\ 
 \hline\hline \end{tabular}} \caption{Mean sex-equal score for six different optimal stable matchings and output from Algorithm REDI.} \label{sm_re_table_sexEqualCostAv} \end{table} 
\clearpage
\begin{table}[] \centerline{\begin{tabular}{ R{1.2cm} | R{1.8cm} R{1.8cm} R{1.8cm} R{1.8cm} R{1.8cm} R{2cm} R{1.8cm} }\hline\hline Case & Balanced & Sex-equal & Egalitarian & Minimum regret & Regret-equal & Min-regret sum & Algorithm REDI  \\ 
\hline S10 & $58.6$ & $59.0$ & $57.7$ & $58.3$ & $58.8$ & $58.9$ & $59.1$ \\ 
 S20 & $170.3$ & $171.7$ & $167.4$ & $169.5$ & $170.7$ & $172.2$ & $172.2$ \\ 
 S30 & $316.9$ & $319.4$ & $311.3$ & $315.5$ & $319.3$ & $322.5$ & $321.8$ \\ 
 S40 & $492.1$ & $495.5$ & $482.7$ & $488.7$ & $493.6$ & $502.4$ & $498.0$ \\ 
 S50 & $691.5$ & $695.0$ & $679.7$ & $687.9$ & $692.6$ & $708.1$ & $698.8$ \\ 
 S60 & $907.2$ & $912.3$ & $892.1$ & $903.2$ & $909.7$ & $930.1$ & $919.4$ \\ 
 S70 & $1149.3$ & $1156.1$ & $1131.7$ & $1145.4$ & $1154.5$ & $1185.7$ & $1163.7$ \\ 
 S80 & $1410.2$ & $1417.1$ & $1389.8$ & $1403.3$ & $1415.2$ & $1452.9$ & $1427.7$ \\ 
 S90 & $1684.8$ & $1693.2$ & $1660.6$ & $1675.1$ & $1687.7$ & $1727.1$ & $1703.8$ \\ 
 S100 & $1973.3$ & $1981.8$ & $1945.9$ & $1963.4$ & $1976.0$ & $2028.7$ & $1992.6$ \\ 
 S200 & $5617.4$ & $5632.0$ & $5559.5$ & $5601.3$ & $5625.2$ & $5790.9$ & $5671.9$ \\ 
 S300 & $10346.3$ & $10368.1$ & $10258.5$ & $10314.5$ & $10358.0$ & $10611.2$ & $10422.3$ \\ 
 S400 & $15937.5$ & $15964.5$ & $15824.3$ & $15932.0$ & $15978.1$ & $16332.6$ & $16075.3$ \\ 
 S500 & $22276.4$ & $22304.7$ & $22135.5$ & $22276.0$ & $22308.6$ & $22800.7$ & $22449.0$ \\ 
 S600 & $29326.0$ & $29366.7$ & $29153.0$ & $29297.7$ & $29371.4$ & $30233.8$ & $29525.9$ \\ 
 S700 & $36929.9$ & $36972.9$ & $36744.2$ & $36933.6$ & $36975.4$ & $37910.1$ & $37167.1$ \\ 
 S800 & $45203.0$ & $45247.6$ & $44977.2$ & $45195.9$ & $45264.7$ & $46423.0$ & $45484.3$ \\ 
 S900 & $53940.0$ & $53986.4$ & $53690.7$ & $54003.9$ & $54053.3$ & $55265.3$ & $54338.2$ \\ 
 S1000 & $63085.6$ & $63139.0$ & $62806.5$ & $63145.7$ & $63204.4$ & $64752.1$ & $63518.5$ \\ 
 \hline\hline \end{tabular}} \caption{Mean cost for six different optimal stable matchings and output from Algorithm REDI.} \label{sm_re_table_egalCostAv} \end{table} 
\clearpage
\begin{table}[] \centerline{\begin{tabular}{ R{1.2cm} | R{1.8cm} R{1.8cm} R{1.8cm} R{1.8cm} R{1.8cm} R{2cm} R{1.8cm} }\hline\hline Case & Balanced & Sex-equal & Egalitarian & Minimum regret & Regret-equal & Min-regret sum & Algorithm REDI  \\ 
\hline S10 & $7.9$ & $7.9$ & $7.8$ & $7.6$ & $7.7$ & $7.8$ & $7.7$ \\ 
 S20 & $14.4$ & $14.5$ & $14.3$ & $13.5$ & $13.8$ & $14.1$ & $13.8$ \\ 
 S30 & $21.0$ & $21.4$ & $20.7$ & $19.3$ & $19.8$ & $20.2$ & $19.9$ \\ 
 S40 & $26.7$ & $27.1$ & $26.1$ & $24.3$ & $25.1$ & $25.7$ & $25.2$ \\ 
 S50 & $31.7$ & $32.0$ & $31.2$ & $28.7$ & $29.5$ & $30.4$ & $29.6$ \\ 
 S60 & $36.6$ & $37.0$ & $35.5$ & $33.1$ & $34.1$ & $35.1$ & $34.2$ \\ 
 S70 & $41.2$ & $41.9$ & $40.6$ & $37.4$ & $38.5$ & $39.8$ & $38.6$ \\ 
 S80 & $45.1$ & $45.6$ & $44.6$ & $41.1$ & $42.4$ & $44.0$ & $42.5$ \\ 
 S90 & $49.1$ & $49.7$ & $48.2$ & $44.6$ & $45.9$ & $47.2$ & $46.1$ \\ 
 S100 & $53.4$ & $54.1$ & $52.0$ & $48.1$ & $49.7$ & $51.2$ & $49.9$ \\ 
 S200 & $87.5$ & $88.0$ & $85.2$ & $79.3$ & $82.0$ & $84.8$ & $82.2$ \\ 
 S300 & $116.5$ & $117.5$ & $114.0$ & $104.3$ & $107.3$ & $110.8$ & $107.3$ \\ 
 S400 & $139.2$ & $139.6$ & $136.8$ & $125.5$ & $129.1$ & $133.4$ & $129.3$ \\ 
 S500 & $160.9$ & $161.2$ & $159.8$ & $146.7$ & $151.1$ & $154.4$ & $151.1$ \\ 
 S600 & $182.1$ & $182.6$ & $178.7$ & $166.0$ & $170.8$ & $177.0$ & $170.9$ \\ 
 S700 & $198.5$ & $198.9$ & $197.4$ & $181.8$ & $186.8$ & $191.9$ & $186.8$ \\ 
 S800 & $217.7$ & $218.6$ & $214.3$ & $197.9$ & $202.2$ & $210.7$ & $202.5$ \\ 
 S900 & $239.6$ & $239.4$ & $235.7$ & $215.6$ & $222.4$ & $228.7$ & $222.5$ \\ 
 S1000 & $252.7$ & $253.7$ & $252.2$ & $232.5$ & $239.5$ & $246.3$ & $239.5$ \\ 
 \hline\hline \end{tabular}} \caption{Mean degree for six different optimal stable matchings and output from Algorithm REDI.} \label{sm_re_table_degreeAv} \end{table} 
\clearpage
\begin{table}[] \centerline{\begin{tabular}{ R{1.2cm} | R{1.8cm} R{1.8cm} R{1.8cm} R{1.8cm} R{1.8cm} R{2cm} R{1.8cm} }\hline\hline Case & Balanced & Sex-equal & Egalitarian & Minimum regret & Regret-equal & Min-regret sum & Algorithm REDI  \\ 
\hline S10 & $1.8$ & $1.8$ & $2.2$ & $1.6$ & $1.5$ & $2.7$ & $1.5$ \\ 
 S20 & $3.2$ & $3.3$ & $4.1$ & $2.6$ & $2.2$ & $5.2$ & $2.2$ \\ 
 S30 & $5.1$ & $5.2$ & $6.3$ & $3.5$ & $2.9$ & $8.2$ & $2.9$ \\ 
 S40 & $6.2$ & $6.3$ & $7.5$ & $4.3$ & $3.2$ & $10.6$ & $3.2$ \\ 
 S50 & $7.2$ & $7.2$ & $8.5$ & $4.7$ & $3.6$ & $12.3$ & $3.6$ \\ 
 S60 & $8.2$ & $8.4$ & $9.4$ & $5.5$ & $4.0$ & $14.1$ & $4.0$ \\ 
 S70 & $9.2$ & $9.5$ & $11.0$ & $6.1$ & $4.0$ & $16.5$ & $4.0$ \\ 
 S80 & $9.5$ & $9.6$ & $11.9$ & $6.4$ & $4.3$ & $18.1$ & $4.3$ \\ 
 S90 & $10.1$ & $10.4$ & $12.0$ & $6.0$ & $4.2$ & $18.1$ & $4.2$ \\ 
 S100 & $11.1$ & $11.2$ & $13.1$ & $7.1$ & $4.7$ & $20.1$ & $4.7$ \\ 
 S200 & $16.5$ & $16.9$ & $18.4$ & $10.7$ & $6.9$ & $33.4$ & $6.9$ \\ 
 S300 & $23.7$ & $24.0$ & $25.7$ & $13.3$ & $8.7$ & $39.6$ & $8.7$ \\ 
 S400 & $25.4$ & $25.7$ & $28.3$ & $14.9$ & $9.0$ & $46.0$ & $9.0$ \\ 
 S500 & $28.0$ & $27.9$ & $30.8$ & $17.9$ & $10.8$ & $52.3$ & $10.8$ \\ 
 S600 & $31.0$ & $30.9$ & $34.5$ & $18.6$ & $11.3$ & $63.7$ & $11.3$ \\ 
 S700 & $32.2$ & $32.3$ & $36.1$ & $20.4$ & $12.2$ & $63.2$ & $12.2$ \\ 
 S800 & $37.0$ & $37.2$ & $39.0$ & $21.1$ & $13.9$ & $71.3$ & $13.9$ \\ 
 S900 & $42.7$ & $41.8$ & $45.5$ & $24.7$ & $14.3$ & $77.0$ & $14.3$ \\ 
 S1000 & $40.4$ & $40.1$ & $45.9$ & $25.9$ & $14.2$ & $84.6$ & $14.2$ \\ 
 \hline\hline \end{tabular}} \caption{Mean regret-equality score for six different optimal stable matchings and output from Algorithm REDI.} \label{sm_re_table_regretEqualScoreAv} \end{table} 
\clearpage
\begin{table}[] \centerline{\begin{tabular}{ R{1.2cm} | R{1.8cm} R{1.8cm} R{1.8cm} R{1.8cm} R{1.8cm} R{2cm} R{1.8cm} }\hline\hline Case & Balanced & Sex-equal & Egalitarian & Minimum regret & Regret-equal & Min-regret sum & Algorithm REDI  \\ 
\hline S10 & $13.8$ & $14.0$ & $13.3$ & $13.1$ & $13.9$ & $12.9$ & $13.9$ \\ 
 S20 & $25.5$ & $25.8$ & $24.5$ & $23.6$ & $25.3$ & $22.9$ & $25.5$ \\ 
 S30 & $36.9$ & $37.6$ & $35.0$ & $33.4$ & $36.8$ & $32.3$ & $36.9$ \\ 
 S40 & $47.2$ & $47.8$ & $44.6$ & $42.3$ & $47.0$ & $40.8$ & $47.1$ \\ 
 S50 & $56.3$ & $56.7$ & $53.7$ & $50.2$ & $55.3$ & $48.4$ & $55.5$ \\ 
 S60 & $64.9$ & $65.6$ & $61.5$ & $57.9$ & $64.2$ & $56.0$ & $64.4$ \\ 
 S70 & $73.3$ & $74.3$ & $70.1$ & $65.6$ & $73.0$ & $63.1$ & $73.1$ \\ 
 S80 & $80.7$ & $81.6$ & $77.0$ & $72.5$ & $80.4$ & $69.8$ & $80.8$ \\ 
 S90 & $88.1$ & $89.0$ & $84.3$ & $79.2$ & $87.6$ & $76.3$ & $87.9$ \\ 
 S100 & $95.7$ & $96.9$ & $90.8$ & $85.4$ & $94.8$ & $82.3$ & $95.1$ \\ 
 S200 & $158.4$ & $159.2$ & $151.7$ & $142.1$ & $157.0$ & $136.3$ & $157.4$ \\ 
 S300 & $209.4$ & $210.9$ & $202.2$ & $187.6$ & $205.8$ & $182.0$ & $206.0$ \\ 
 S400 & $252.9$ & $253.4$ & $245.3$ & $227.6$ & $249.2$ & $220.7$ & $249.6$ \\ 
 S500 & $293.8$ & $294.5$ & $288.8$ & $265.2$ & $291.4$ & $256.5$ & $291.5$ \\ 
 S600 & $333.1$ & $334.2$ & $322.8$ & $300.9$ & $330.3$ & $290.3$ & $330.5$ \\ 
 S700 & $364.8$ & $365.6$ & $358.8$ & $330.0$ & $361.3$ & $320.6$ & $361.4$ \\ 
 S800 & $398.3$ & $400.0$ & $389.6$ & $363.3$ & $390.5$ & $350.0$ & $391.1$ \\ 
 S900 & $436.5$ & $436.9$ & $426.0$ & $391.4$ & $430.5$ & $380.3$ & $430.8$ \\ 
 S1000 & $465.0$ & $467.3$ & $458.5$ & $420.9$ & $464.7$ & $408.0$ & $464.9$ \\ 
 \hline\hline \end{tabular}} \caption{Mean regret sum for six different optimal stable matchings and output from Algorithm REDI.} \label{sm_re_table_sumRegretAv} \end{table}

\FloatBarrier
\clearpage
\newpage
\subsection{Additional experiments and evaluations}
\label{sm_re_exps_sec_additional_app}

This appendix section describes further experiments and evaluations not referred to in Section \ref{sm_re_exps_sec}.

A summary of generated instance information may be seen in Table \ref{sm_re_table_instinfo}. As in Section \ref{sm_re_exps_sec_f_t_app}, instance types are labelled according to $n$, e.g., $S100$ is the instance type containing instances where $n=100$. Columns $3$ and $4$ show the mean number of stable matchings $|\mathcal{M}|_{av}$ and mean number of rotations $|R|_{av}$, respectively. 
Figure \ref{sm_re_fig_mean_num_optimal} (associated with Table \ref{sm_re_table_mean_num_optimal}) shows a bar chart of the mean number of stable matchings occurring for the six different types of optimal stable matching described above, with increasing $n$. Finally, Figure \ref{sm_re_fig_mean_num_matchings_num_opt} (associated with Table \ref{sm_re_table_mean_num_matchings_num_opt}) shows a bar chart of the mean number of stable matchings that satisfy different numbers of optimal stable matching criteria, with increasing $n$. Both of these bar charts show a reduced selection of $n$ with $n \in \{100, 400, 700, 1000\}$ (the tables show the full data).

We note the following additional results:

\begin{itemize}
	\item \emph{Balanced and sex-equal stable matchings:} In all plots of Figure \ref{reg-eq-opt-plots}, balanced and sex-equal stable matchings have remarkably similar mean scores over all instance sizes and all measures. This may be due to the similar nature of these optimality measures, where both measures involve a calculation over the cost of matchings (recall that the balanced objective involves minimising the maximum of the total cost for the men and the total cost for the women, whilst the sex-equality objective involves minimising the absolute value of the difference between the total cost for the men and the total cost for the women). This similarity was previously noted by Manlove \cite[pg. 110]{Man13}, who references work undertaken by Eric McDermid to find an instance of \acrshort{smi} in which no balanced stable matching is also a sex-equal stable matching.

	\item \emph{Frequency of different types of optimal stable matchings:} From Figure \ref{sm_re_fig_mean_num_optimal}, we can see a clear ordering for the three most frequent types of degree-based optimal stable matching. From most frequent to least frequent they are minimum regret, regret-equal and min-regret sum stable matchings. The minimum-regret stable matching may be most frequent because this optimality criterion is somewhat less constrained than the other two, as it is based only on the worst performing agent. In contrast, the regret-equal and min-regret sum stable matchings are based on the worst performing man and worst performing woman. Additionally, cost-based optimal stable matchings are likely to be more constrained than degree-based ones (due to the number of different costs and degrees possible for stable matchings of any $n$), which may account for the very low average number of stable matchings for these types.

	\item \emph{The number of optimality criteria that stable matchings satisfy:} The bar chart in Figure \ref{sm_re_fig_mean_num_matchings_num_opt} shows a clear pattern of fewer stable matchings satisfying higher numbers of optimality criteria. For smaller instances there is a trend for this to level out somewhat, which can be seen more clearly for the instance types with lowest $n$ in Table \ref{sm_re_table_mean_num_matchings_num_opt}. Taken to extreme, if there is only one stable matching in an instance, it will necessarily satisfy all optimality criteria. Thus as $n$ increases, so too does the number of stable matchings and it is therefore less likely that higher numbers of optimality criteria are satisfied by a particular stable matching.
\end{itemize}

\begin{table}[] \centerline{\begin{tabular}{ R{1.2cm} | R{1.6cm} R{1.6cm} R{1.6cm} }\hline\hline Case & $n$ & $|\mathcal{M}|_{av}$ & $|R|_{av}$ \\ 
\hline S10 & $10$ & $3.0$ & $1.8$ \\ 
 S20 & $20$ & $6.7$ & $4.2$ \\ 
 S30 & $30$ & $10.3$ & $6.3$ \\ 
 S40 & $40$ & $16.1$ & $8.9$ \\ 
 S50 & $50$ & $21.0$ & $11.2$ \\ 
 S60 & $60$ & $27.7$ & $13.7$ \\ 
 S70 & $70$ & $33.1$ & $15.8$ \\ 
 S80 & $80$ & $40.8$ & $18.1$ \\ 
 S90 & $90$ & $48.0$ & $20.4$ \\ 
 S100 & $100$ & $54.8$ & $22.7$ \\ 
 S200 & $200$ & $139.2$ & $42.4$ \\ 
 S300 & $300$ & $219.1$ & $58.9$ \\ 
 S400 & $400$ & $348.4$ & $76.2$ \\ 
 S500 & $500$ & $442.5$ & $90.3$ \\ 
 S600 & $600$ & $546.2$ & $105.7$ \\ 
 S700 & $700$ & $670.5$ & $118.8$ \\ 
 S800 & $800$ & $815.2$ & $132.5$ \\ 
 S900 & $900$ & $977.0$ & $144.0$ \\ 
 S1000 & $1000$ & $1077.5$ & $156.7$ \\ 
 \hline\hline \end{tabular}} \caption{General instance information.} \label{sm_re_table_instinfo} \end{table}

\clearpage
\newpage

\begin{figure}
	\centering
    \includegraphics[scale=0.8]{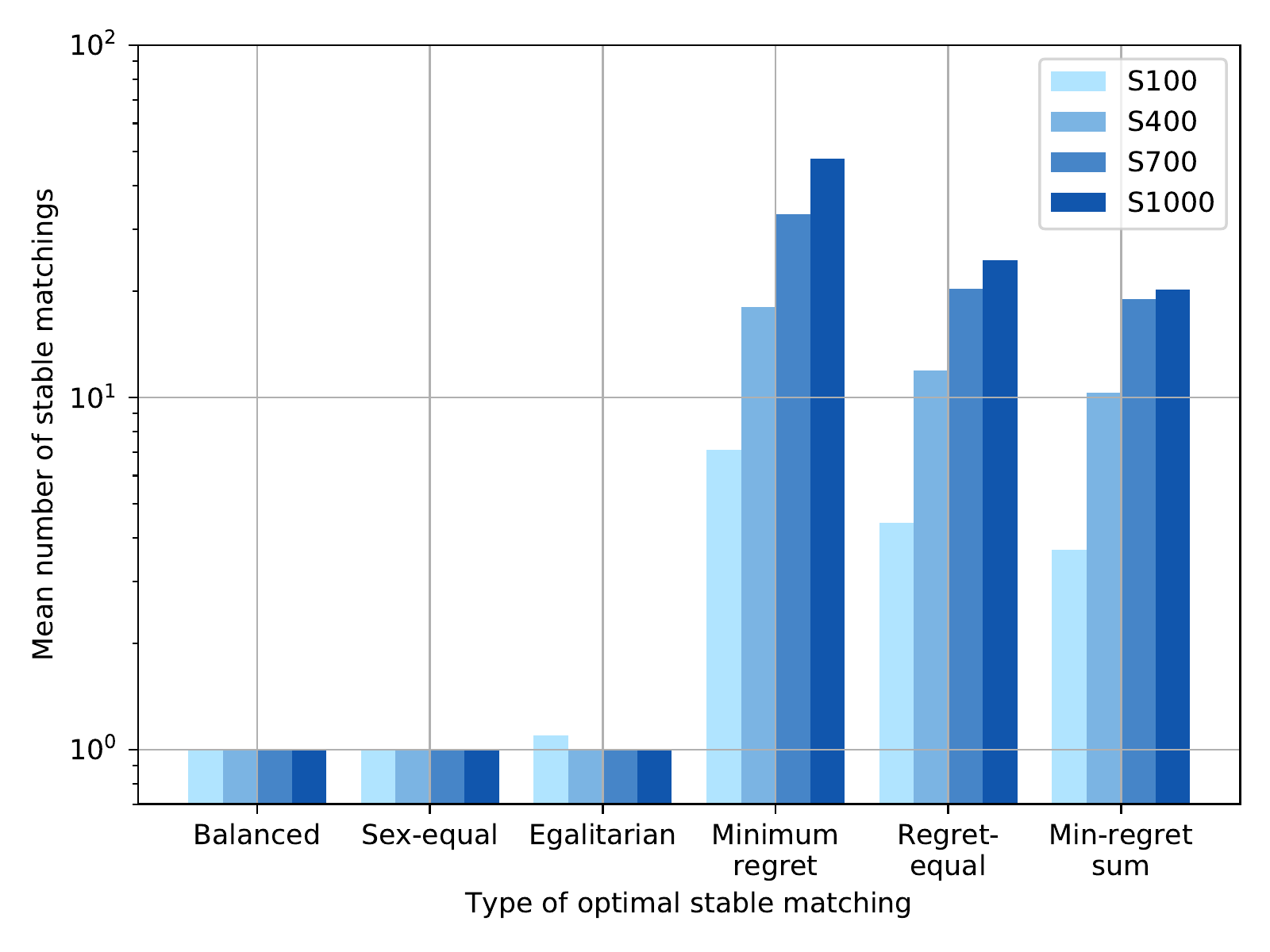}
    \caption{Bar chart of the mean number of stable matchings for different types of optimal matchings, for $n \in \{100, 400, 700, 1000\}$.}
    \label{sm_re_fig_mean_num_optimal}
\end{figure}

\begin{figure}
	\centering
    \includegraphics[scale=0.8]{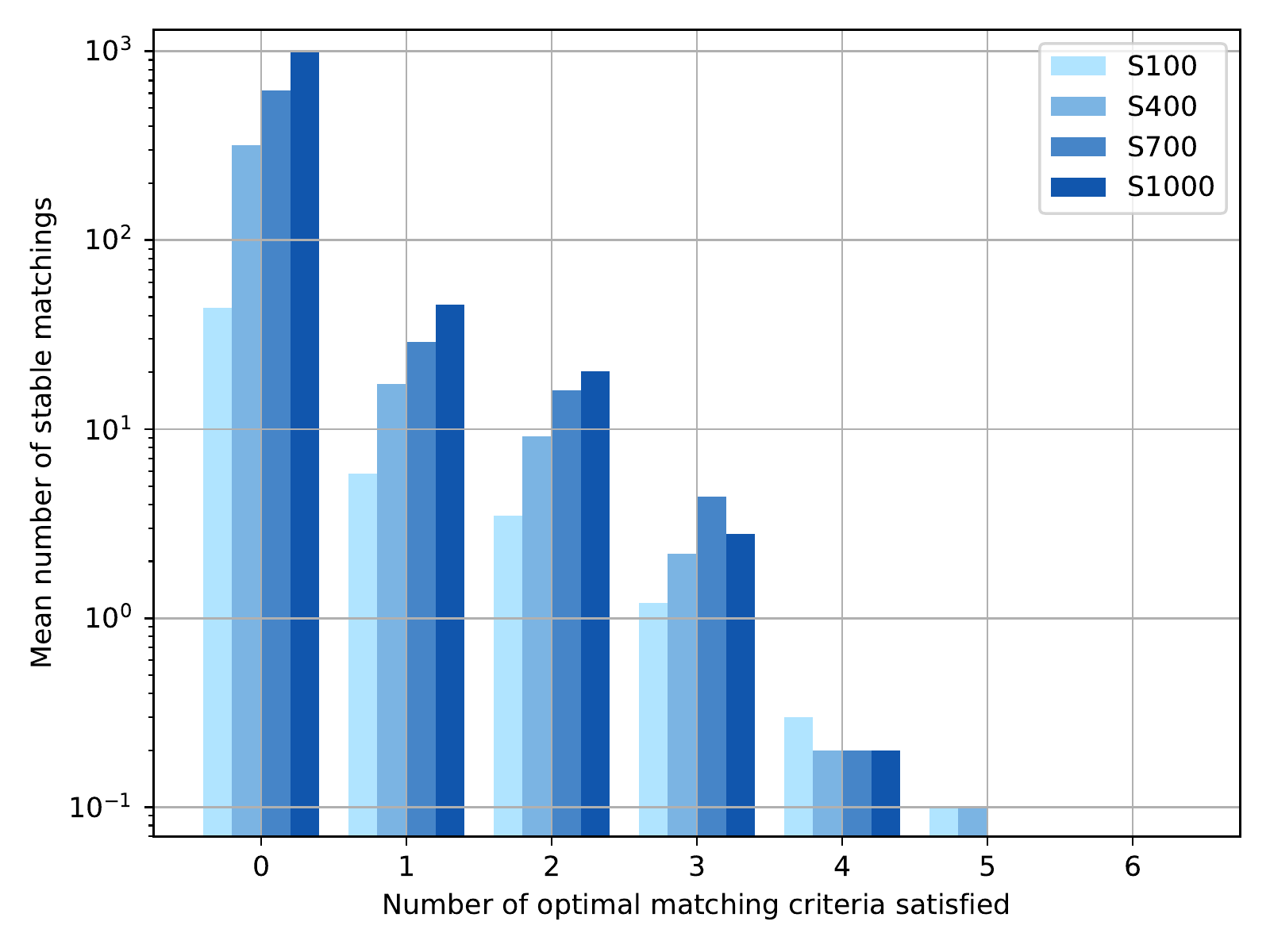}
    \caption{Bar chart of the mean number of stable matchings which satisfy different numbers of optimal stable matching criteria, for $n \in \{100, 400, 700, 1000\}$.}
    \label{sm_re_fig_mean_num_matchings_num_opt}
\end{figure}

\begin{table}[] \centerline{\begin{tabular}{ R{1.2cm} | R{2cm} R{2cm} R{2cm} R{2cm} R{2cm} R{2cm} }\hline\hline Case & Balanced & Sex-equal & Egalitarian & Minimum regret & Regret-equal & Min-regret sum  \\ 
\hline S10 & $1.1$ & $1.0$ & $1.1$ & $1.6$ & $1.3$ & $1.4$ \\ 
 S20 & $1.0$ & $1.0$ & $1.1$ & $2.1$ & $1.7$ & $1.7$ \\ 
 S30 & $1.0$ & $1.0$ & $1.1$ & $2.7$ & $1.9$ & $1.8$ \\ 
 S40 & $1.0$ & $1.0$ & $1.1$ & $3.3$ & $2.3$ & $2.1$ \\ 
 S50 & $1.0$ & $1.0$ & $1.1$ & $3.9$ & $2.7$ & $2.4$ \\ 
 S60 & $1.0$ & $1.0$ & $1.1$ & $4.7$ & $3.3$ & $3.0$ \\ 
 S70 & $1.0$ & $1.0$ & $1.0$ & $4.9$ & $3.3$ & $2.8$ \\ 
 S80 & $1.0$ & $1.0$ & $1.1$ & $5.4$ & $3.5$ & $3.3$ \\ 
 S90 & $1.0$ & $1.0$ & $1.0$ & $6.3$ & $4.3$ & $3.3$ \\ 
 S100 & $1.0$ & $1.0$ & $1.1$ & $7.1$ & $4.4$ & $3.7$ \\ 
 S200 & $1.0$ & $1.0$ & $1.0$ & $11.2$ & $7.1$ & $6.3$ \\ 
 S300 & $1.0$ & $1.0$ & $1.0$ & $16.3$ & $9.7$ & $8.8$ \\ 
 S400 & $1.0$ & $1.0$ & $1.0$ & $18.1$ & $11.9$ & $10.3$ \\ 
 S500 & $1.0$ & $1.0$ & $1.0$ & $23.7$ & $15.8$ & $13.1$ \\ 
 S600 & $1.0$ & $1.0$ & $1.0$ & $30.0$ & $17.5$ & $15.2$ \\ 
 S700 & $1.0$ & $1.0$ & $1.0$ & $33.2$ & $20.3$ & $19.0$ \\ 
 S800 & $1.0$ & $1.0$ & $1.0$ & $35.1$ & $21.2$ & $16.4$ \\ 
 S900 & $1.0$ & $1.0$ & $1.0$ & $48.6$ & $28.2$ & $21.2$ \\ 
 S1000 & $1.0$ & $1.0$ & $1.0$ & $47.7$ & $24.5$ & $20.2$ \\ 
 \hline\hline \end{tabular}} \caption{Mean number of optimal stable matchings per instance.} \label{sm_re_table_mean_num_optimal} \end{table} 
\begin{table}[] \centerline{\begin{tabular}{ R{1.2cm} | R{1.5cm} R{1.5cm} R{1.5cm} R{1.5cm} R{1.5cm} R{1.5cm} R{1.5cm} }\hline\hline Case & $0$ & $1$ & $2$ & $3$ & $4$ & $5$ & $6$  \\ 
\hline S10 & $0.7$ & $0.4$ & $0.5$ & $0.4$ & $0.3$ & $0.3$ & $0.4$ \\ 
 S20 & $3.2$ & $1.1$ & $1.0$ & $0.6$ & $0.3$ & $0.2$ & $0.2$ \\ 
 S30 & $5.9$ & $1.7$ & $1.4$ & $0.6$ & $0.4$ & $0.2$ & $0.1$ \\ 
 S40 & $10.6$ & $2.3$ & $1.8$ & $0.8$ & $0.4$ & $0.2$ & $0.1$ \\ 
 S50 & $14.7$ & $2.8$ & $2.1$ & $0.9$ & $0.4$ & $0.2$ & $0.0$ \\ 
 S60 & $20.0$ & $3.6$ & $2.6$ & $1.0$ & $0.3$ & $0.1$ & $0.0$ \\ 
 S70 & $25.0$ & $4.1$ & $2.7$ & $0.9$ & $0.3$ & $0.1$ & $0.0$ \\ 
 S80 & $32.0$ & $4.2$ & $3.3$ & $0.9$ & $0.3$ & $0.1$ & $0.0$ \\ 
 S90 & $37.9$ & $5.1$ & $3.6$ & $0.9$ & $0.3$ & $0.1$ & $0.0$ \\ 
 S100 & $43.8$ & $5.8$ & $3.5$ & $1.2$ & $0.3$ & $0.1$ & $0.0$ \\ 
 S200 & $121.5$ & $9.9$ & $6.0$ & $1.4$ & $0.3$ & $0.1$ & $0.0$ \\ 
 S300 & $194.4$ & $14.1$ & $8.4$ & $1.8$ & $0.3$ & $0.1$ & $0.0$ \\ 
 S400 & $319.5$ & $17.3$ & $9.2$ & $2.2$ & $0.2$ & $0.1$ & $0.0$ \\ 
 S500 & $404.3$ & $23.5$ & $12.5$ & $2.0$ & $0.2$ & $0.0$ & $0.0$ \\ 
 S600 & $501.3$ & $26.4$ & $16.3$ & $1.8$ & $0.3$ & $0.0$ & $0.0$ \\ 
 S700 & $620.7$ & $29.1$ & $16.0$ & $4.4$ & $0.2$ & $0.0$ & $0.0$ \\ 
 S800 & $765.0$ & $28.8$ & $17.7$ & $3.4$ & $0.3$ & $0.0$ & $0.0$ \\ 
 S900 & $905.8$ & $44.5$ & $23.8$ & $2.7$ & $0.2$ & $0.0$ & $0.0$ \\ 
 S1000 & $1008.5$ & $45.8$ & $20.3$ & $2.8$ & $0.2$ & $0.0$ & $0.0$ \\ 
 \hline\hline \end{tabular}} \caption{Mean number of stable matchings that satisfy $c$ optimality criteria, where $c$ varies on the x-axis.} \label{sm_re_table_mean_num_matchings_num_opt} \end{table}

\end{document}